\PassOptionsToPackage{dvipsnames}{xcolor}
\documentclass[10pt,sigplan,screen,nonacm]{acmart}
\renewcommand\footnotetextcopyrightpermission[1]{} 
\usepackage[utf8]{inputenc}
\usepackage[T1]{fontenc}
\usepackage[autostyle]{csquotes}
\usepackage{mathtools}
\usepackage{algorithm}
\usepackage{amssymb}
\usepackage[noend]{algpseudocode}
\usepackage{enumitem}
\usepackage{varwidth}
\usepackage{caption}
\usepackage{subcaption}
\usepackage{booktabs}
\usepackage{adjustbox}
\usepackage{tikz}
\usepackage{pgfplots}
\usepackage{balance}
\usetikzlibrary{positioning}
\usetikzlibrary{shapes}
\usetikzlibrary{calc}
\usetikzlibrary{shapes.misc}
\usetikzlibrary{patterns}
\usepackage{xcolor}
\usepackage{xspace}

\algblockdefx[UPON]{Upon}{EndUpon}[2][Unknown]{\textbf{upon reception of} \textsc{#1}$\langle$#2$\rangle$}{}
\algblockdefx[UPON]{Wait}{EndWait}[2][Unknown]{\textbf{wait until either}}{}
\makeatletter
\algrenewcommand\ALG@beginalgorithmic{\footnotesize}
\makeatother
\definecolor{darkgray}{gray}{0.40}
\newcommand{\InlineComment}[1]{\State{\color{darkgray} \scriptsize $\triangleright$ #1}}
\algrenewcommand{\algorithmiccomment}[1]{\hfill{}{\color{darkgray} \scriptsize $\triangleright$ #1}}

\tikzset{every mark/.append style={mark size=2}}
\usepackage{etoolbox}
\usepackage{tikz}
\usetikzlibrary{tikzmark}
\usetikzlibrary{calc}

\errorcontextlines\maxdimen

\newcommand{\ALGtikzmarkcolor}{black}
\newcommand{\ALGtikzmarkextraindent}{4pt}
\newcommand{\ALGtikzmarkverticaloffsetstart}{-.5ex}
\newcommand{\ALGtikzmarkverticaloffsetend}{-.5ex}
\newcommand{\ALGtikzmarkhorizontaloffset}{3pt}
\makeatletter
\newcounter{ALG@tikzmark@tempcnta}

\newcommand\ALG@tikzmark@start{%
    \global\let\ALG@tikzmark@last\ALG@tikzmark@starttext%
    \expandafter\edef\csname ALG@tikzmark@\theALG@nested\endcsname{\theALG@tikzmark@tempcnta}%
    \tikzmark{ALG@tikzmark@start@\csname ALG@tikzmark@\theALG@nested\endcsname}%
    \addtocounter{ALG@tikzmark@tempcnta}{1}%
}

\def\ALG@tikzmark@starttext{start}
\newcommand\ALG@tikzmark@end{%
    \ifx\ALG@tikzmark@last\ALG@tikzmark@starttext
    \else
        \tikzmark{ALG@tikzmark@end@\csname ALG@tikzmark@\theALG@nested\endcsname}%
        \tikz[overlay,remember picture] \draw[\ALGtikzmarkcolor] let \p{S}=($(pic cs:ALG@tikzmark@start@\csname ALG@tikzmark@\theALG@nested\endcsname)+(\ALGtikzmarkextraindent,\ALGtikzmarkverticaloffsetstart)$), \p{E}=($(pic cs:ALG@tikzmark@end@\csname ALG@tikzmark@\theALG@nested\endcsname)+(\ALGtikzmarkextraindent,\ALGtikzmarkverticaloffsetend)$), \p{F}=($(\p{S})+(\ALGtikzmarkhorizontaloffset,0pt)$) in (\x{S},\y{S})--(\x{S},\y{E})--(\x{F},\y{E});%
    \fi
    \gdef\ALG@tikzmark@last{end}%
}

\apptocmd{\ALG@beginblock}{\ALG@tikzmark@start}{}{\errmessage{failed to patch}}
\pretocmd{\ALG@endblock}{\ALG@tikzmark@end}{}{\errmessage{failed to patch}}
\makeatother

\makeatletter
\ifthenelse{\equal{\ALG@noend}{t}}%
  {\algtext*{EndUpon}}
  {}%
\makeatother

\newcommand{\BFTSMART}[1][]{\textsc{BFT-SMaRt}\xspace}
\newcommand{\id}[1][]{\mathit{id}}

\renewcommand{\Statex}{\vspace{0.5em}}

\newcommand{\proj}{PnyxDB}

\title{{\proj{}}:~a {Lightweight} {Leaderless} {Democratic} {Byzantine Fault Tolerant} {Replicated} {Datastore}}

\author{Loïck Bonniot}
\affiliation{\institution{InterDigital}}
\affiliation{\institution{Univ Rennes, Inria, CNRS, IRISA}\city{Rennes}\country{France}}
\email{loick.bonniot@interdigital.com}

\author{Christoph Neumann}
\affiliation{\institution{InterDigital}\city{Rennes}\country{France}}
\email{christoph.neumann@interdigital.com}

\author{François Taïani}
\affiliation{\institution{Univ Rennes, Inria, CNRS, IRISA}\city{Rennes}\country{France}}
\email{francois.taiani@irisa.fr}

\begin{abstract}
  Byzantine-Fault-Tolerant (BFT) systems are rapidly emerging as a viable technology for production-grade systems, notably in closed consortia deployments for financial and supply-chain applications.
Unfortunately, most algorithms proposed so far to coordinate these systems suffer from substantial scalability issues, and lack important features to implement Internet-scale governance mechanisms.

In this paper, we observe that many application workloads offer little concurrency, and propose \proj{}, an eventually-consistent Byzantine Fault Tolerant replicated datastore that exhibits both high scalability and low latency.
Our approach is based on conditional endorsements, that allow nodes to specify the set of transactions that must \emph{not} be committed for the endorsement to be valid.
In addition to its high scalability, \proj{} supports application-level voting, i.e. individual nodes are able to endorse or reject a transaction according to application-defined policies without compromising consistency.
We provide a comparison against \BFTSMART and Tendermint, two competitors with different design aims, and show that our implementation speeds up commit latencies by a factor of 11, remaining below 5 seconds in a worldwide geodistributed deployment of 180 nodes.

\end{abstract}

\begin{document}
\maketitle{}

\section{Introduction}

Byzantine-Fault-Tolerant (BFT) systems have attracted a large body of works over the last two decades~\cite{Malkhi1998,Malkhi1998a,Castro1999,Bessani2008,Bessani2014,Kotla2009,Duan2014,Next700BFT},
and have now moved into the public spotlight following the dramatic rise of blockchain platforms~\cite{Nakamoto2008,EthereumFoundation}.
These systems typically rely on powerful BFT replication protocols to ensure consistency between their replicas, and withstand arbitrary failures and potential malicious behavior.
Unfortunately, traditional BFT replication protocols struggle to scale beyond a few tens of replicas~\cite{Dantas2007}, while the proof-of-work technique used by many blockchain-based systems suffers from large computing and storage overheads.

Recent attempts to overcome these scalability barriers have explored leaderless designs~\cite{Abd-El-Malek2005,Pedone2011,Luiz2011,Lamport2011,Moraru2013,Crain2017,Rocket2019},
alternatives to proof-of-work such as proof-of-stake~\cite{Gilad2017}, or assumed access to a trusted third party providing strong coordination and ordering guarantees~\cite{Androulaki2018}.
All these strategies are however fraught with limitations: existing leaderless protocols rely either on clients for consistency checks~\cite{Abd-El-Malek2005} (increasing computing overhead)
or on the availability of strong coordination mechanisms, such as a trusted peer-sampling service~\cite{Rocket2019} or atomic broadcast primitives~\cite{Pedone2011,Luiz2011,Crain2017,Androulaki2018};
proof-of-stake links a node's influence to its stake in the system,
a problematic dependency for many use cases; and trusted third parties considerably limit the applicability of such solutions to well-controlled environments.

Compounding these limitations, all above approaches are ill-equipped to support \emph{in-system governance mechanisms}, a growing requirement for applications involving independent organizations~\cite{Goodman2014}.
More specifically, although most of these solutions rely on internal voting or quorum mechanisms, these mechanisms are not exposed to applications as first-class primitives.
As a result, individual nodes cannot implement application-defined policies to endorse or reject transactions without additional effort, costs, and complexity.
This is problematic, as such application-level voting capabilities are key to a number of emerging decentralized BFT applications
involving independent participants who need to balance conflicting goals and shared interests~\cite{Noveck2009,EthereumFoundation}.
Examples of such governance concerns include basic membership management with access control, resource allocation and sharing, crowdsourced scheduling, policy administration and knowledge distribution.
In all these examples, different parties are likely to pursue different agendas, prompting the need for participants~to be able to influence the distributed decision making process according to their own application-defined policies and beliefs~\cite{Goodman2014,Cai2018,Malkhi2019}.

To address these challenges, we advocate in this paper a radically different line of attack: we borrow a popular strategy from non-Byzantine distributed datastores~\cite{Sutra2006,Vogels2008,Shapiro2011,Guerraoui2019},
and tackle scalability by weakening the consistency guarantees, while maintaining Byzantine Fault Tolerance.
We illustrate this design with \emph{\proj{}}\footnote{The Pnyx hill was used as the main meeting place in Athenian democracy.}
 a \emph{Byzantine-Fault-Tolerant Replicated Datastore for closed consortia}.
\proj{} is \emph{eventually consistent} in that clients might perceive conflicting views of the datastore for short periods of time.
\proj{} also provides a unique application-level voting mechanism that allow participating nodes to support or reject proposed transactions according to application-defined policies.

Our proposal leverages the long-observed fact that many workloads exhibit a large proportion of commutative and independent operations~\cite{Moraru2013,Lu2015}:
these operations can be executed out of order without compromising the eventual convergence of all correct nodes.
We exploit these commutative operations through a \emph{modified Byzantine Quorum protocol}~\cite{Malkhi1998a} that ensures the safety and agreement of our system.
More specifically, we introduce \emph{conditional endorsements} within quorums as a mean to flag and handle conflicts by allowing each node to specify the set of transactions that must \emph{not} be committed for the endorsement to be valid.

In this paper, we make the following contributions:
\begin{enumerate}
  \item We present \emph{\proj{}}, a scalable low-latency BFT replicated datastore that supports democratic voting for participants.
  \item We propose a novel conflict resolution protocol that is resilient to Byzantine faults.
    This protocol lies at the heart of \proj{}, and leverages commutative and independent operations to ensure safety in the face of Byzantine behavior, while delivering scalability and low-latency.
  \item We implemented \proj{} and published its source code~\cite{PnyxDBSourceCode}
    We evaluate our implementation against two well-known systems, \BFTSMART~\cite{Bessani2008,Bessani2014} and Tendermint~\cite{Buchman2016}, two competitors representing alternative trade-offs in the design space.
    We demonstrate that our system is able to reduce commit latencies by at least an order of magnitude under realistic Internet conditions, while maintaining steady commit throughput.
    We also show that \proj{} is able to scale to up to 180 replicas on a worldwide geodistributed AWS deployment, with an average latency of a few seconds.
\end{enumerate}

The remainder of this paper is structured as follows.
Sections~\ref{sec:overview} and~\ref{sec:protocol} define our model and specifies our replication protocol, alongside with properties proofs.
Section~\ref{sec:implementation} presents the technical choices made to implement \proj{}.
Section~\ref{sec:evaluation} evaluates our \proj{} implementation.
We present related work in section~\ref{sec:related_work}, followed by a discussion of limitations in section~\ref{sec:discussion}.
Section~\ref{sec:conclusion} concludes this paper.

\section{\proj{} overview\label{sec:overview}}

\subsection{System Model and Assumptions}

We assume a system made of distributed machines (\emph{nodes}) communicating through messages.
Our system defines three types of roles that one node may implement independently:
\begin{itemize}
  \item \textbf{Clients} can submit transactions, each consisting of a list of operations on a replicated key-value datastore;
  \item \textbf{Endorsers} are able to participate in Byzantine consensus quorums by validating and voting on clients' transactions.
    Like existing decentralized ledgers, they store the whole datastore state in order to serve clients and make policy-based decisions;
  \item and \textbf{Observers} maintain a copy of the shared datastore, but are not able to validate transactions.
\end{itemize}
Each system contains a known number $n$ of endorsers, of which a maximum of $f$ can act as \emph{Byzantine}.
Byzantine nodes are allowed to ignore the protocol specification occasionally or completely, and they can collude to create more sophisticated attacks.
Such a behavior is typically the case for malformed, corrupted or malicious nodes.
Non-Byzantine nodes are said to be \textit{correct}.

We also assume we have access to a reliable BFT broadcast primitive with the following property: \emph{if one message is delivered to one correct node, every correct node will eventually receive that message}~\cite{Bracha1987}.
In our implementation, we rely on eventually synchronous networks to ensure that assumption, as detailed in~\autoref{sec:gossip}.
Cryptographic signatures are used to verify nodes' identity and authorizations. We make the standard assumptions that Byzantine participants cannot break these signatures, and that participants know each other beforehand.
In the parlance of distributed ledgers, our system is \emph{permissioned}: this allows for message authenticity and
data access control while staying relatively dynamic.

\subsection{Intuition and Overview}

Closed-membership Byzantine state machine replication typically rely on some form of Byzantine-tolerant consensus that ensures strong consistency~\cite{,Bessani2014,Next700BFT,Buchman2016,Sousa2018}.
As a result,  they unfortunately do not scale beyond a few tens of replicated nodes, due to the inherent cost of executing a Byzantine agreement protocol~\cite{Dutta2005,Martin2006}.
One strategy to overcome this scalability barrier exploits a trusted computing base for coordination and ordering, such as Kafka or Raft in recent versions of Hyperledger~\cite{TheOrderingService,Androulaki2018},
but this approach weakens the security model of the protocol.
Another strategy consists in using proof-of-work or proof-of-stake techniques from open-membership Byzantine ledgers~\cite{Ball2017,Abraham2018,Gilad2017}.
These techniques are either costly or link a node's influence to its stake in the system, two undesirable properties in many cases.
In this paper we tackle scalability by weakening consistency guarantees---a strategy often used by large-scale datastores---while maintaining Byzantine Fault Tolerance (BFT).

\begin{figure}
	\centering
	\input{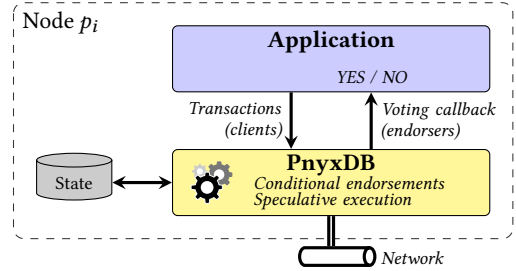}
	\caption{Overview of \proj{}: the application submits transactions to be executed on shared state, and polls the application back for transaction approval before creating conditional endorsements.\label{fig:overview}}
\end{figure}

\autoref{fig:overview} gives an overview of \proj{}'s interface and mechanisms. 
Clients submit transactions that are made of \emph{operations} on keys of the \proj{} datastore.
These operations are typically reads and writes, but \proj{} can be extended to other shared objects with a sequential specification.
These transactions are then broadcast to all endorser nodes, which vote for or against the transaction through an application-level \emph{voting callback}.
This callback provides \emph{in-system governance} by allowing nodes to endorse transactions according to application-level policies. 
Transactions must be supported by a configurable lower threshold of a majority of correct nodes to proceed.

The properties of \proj{} result from the novel combination of two key ingredients: \emph{leaderless quorums} for scalability, and \emph{conditional endorsements} for eventual consistency.

\subsubsection{Leaderless quorums} \proj{} does not use any coordinator, rotating or elected, in contrast to many existing BFT replication solutions~\cite{Buchman2016,Bessani2014,Next700BFT,Castro1999}.
This choice removes a recurring performance bottleneck~in the process, trading off weaker consistency guaranties for higher scalability.
Transactions only need to be endorsed by a Byzantine quorum of endorsers (more than $\frac{n+f}{2}$) to be permanently committed to the system's state.
If two transactions commute (i.e. they contain no conflicting operations), their respective quorums can be built independently, and the transactions applied out of order, thus ensuring \proj{}'s eventual consistency.
This strategy is directly inspired from Conflict-Free Replicated Datatypes (CRDTs)~\cite{Shapiro2011,Raykov2011,Sutra2006} and leverages the fact that many operations in distributed datastores either commute or are independent. 
When this is the case, these transactions may be executed out of order on different nodes without breaking local consistency~\cite{Raykov2011,Li2012}, while allowing every correct node to eventually converge to the same global datastore state.
A typical example is the popular Unspent Transaction Outputs model (UTXO) used in cryptocurrencies~\cite{Nakamoto2008,Frey2016} that avoids concurrency by writing to a variable only once: within this model, conflicts only occur when Byzantine nodes try to re-use an expired variable.
(This problem is well-known as the ``double-spending'' attack.)

\subsubsection{Conditional endorsements} Leaderless quorums work well for commutative transactions, but might lead to deadlocks in case of conflicts, for instance when modifying the same key with non-commutative operations.
We overcome this problem with a second core mechanism: \emph{conditional endorsements}.
When an endorser broadcasts an endorsement, it also publishes a (possibly empty) list of transactions that must \emph{not} be committed for the endorsement to be valid.
(These conflicting transactions are the \emph{conditions} of the endorsement.)
Given a pair of conflicting transactions, all correct nodes will use the same heuristics (based on time-stamps) to decide which one to promote over the other, ensuring a consistent conflict resolution.
Without additional mechanisms, conditional endorsements may however lead to an ever-growing acyclic dependency graph between transactions.
We avoid this outcome by periodically triggering garbage collections (or \emph{checkpoints}) using a binary \emph{Byzantine Veto Procedure} (\autoref{sec:checkpoints_eval}).

\begin{figure}[tb]
  \centering
  \begin{tikzpicture}[
  node distance=6mm,
  state/.style={%
    rectangle,
    minimum size=5mm,
    rounded corners=1mm,
    very thick,
    draw=black!60,
    font=\scshape\small,
  },
  arrow/.style={very thick,>=stealth},
  label/.style={auto,font=\scriptsize\itshape},
  ]
  \node (init) [circle,draw=black,fill=black,minimum size=1mm,inner sep=0mm] {};
  \node (pending)    [state,below=of init,fill=black!20,yshift=3mm] {Pending};
  \node (applicable) [state,below=of pending] {Applicable};
  \node (applied)    [state,below=of applicable] {Applied};
  \node (committed)  [state,below=of applied,fill=teal!20,draw=teal!60] {Committed};
  \node (dropped)    [state,right=of committed,xshift=12mm,fill=red!20,draw=red!60] {Dropped};

  \draw [->,arrow] (init) to node [label,yshift=2mm] {Client submission} (pending);
  \draw [->,arrow,bend left,out=45] (pending.east) to node [label] {Cleanup} (dropped);
  \draw [->,arrow,bend left] (pending) to node [label] {Endorsements} (applicable);
  \draw [<-,arrow,bend right] (pending) to (applicable);
  \draw [->,arrow] (applicable) to node [label] {Speculative Execution} (applied);
  \draw [->,arrow] (applied) to node [label] {Commit} (committed);
  \draw [->,arrow,bend left,out=90,in=90] (applied) to node [label] {Rollback} (pending);
\end{tikzpicture}
  \caption{Transaction state diagram, as viewed by a node.
  From the \textsc{Pending} state, a transaction evolve either to \textsc{Dropped} or \textsc{Committed} given received messages. \textsc{Dropped} and \textsc{Committed} are eventually consistent across all nodes. In contrast, \textsc{Pending}, \textsc{Applicable} and \textsc{Applied} are intermediate states local to each node.}\label{fig:tx_state}
\end{figure}
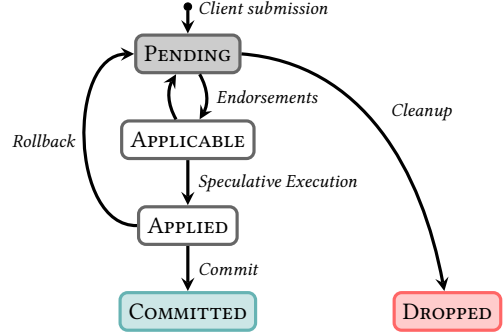

\begin{table}[t]
  \footnotesize
  \renewcommand{\arraystretch}{1.1}
  \caption{Notations used in this paper.}\label{table:notations}
  \centering
  \begin{tabular}{rl}
  \toprule
  \multicolumn{2}{l}{\textbf{System parameters}} \\
  $n$      & Number of nodes \\
  $f$      & Number of faulty nodes \\
  $\omega$ & Required quorum of endorsements \\
  $\Gamma(\Delta, \bar{\Delta})$ &
    $= \left\{\begin{matrix*}[l]
      \textbf{true}~\emph{if}~\Delta~\emph{and}~\bar{\Delta}~\emph{conflict}
      \\
      \textbf{false}~\emph{otherwise}
    \end{matrix*}\right.$ \\
    & \hspace{5mm}where $\Delta, \bar{\Delta}$ are two lists of operations \\
  \midrule
  \multicolumn{2}{l}{\textbf{Message $t \gets \textsc{Transaction}\langle \id, d, R, \Delta \rangle$}} \\
  $t.\id$      & Unique identifier \\
  $t.d$      & Absolute deadline \\
  $t.R$      & Preconditions on datastore state \\
  $t.\Delta$ & List of operations \\
  \midrule
  \multicolumn{2}{l}{\textbf{Message $e \gets \textsc{Endorsement}\langle \id, i, C \rangle$}} \\
  $e.\id$ & Endorsed transaction unique identifier \\
  $e.i$ & Endorser node identifier \\
  $e.C$ & Endorsement conditions, the set of transactions that \\
      & must \textbf{not} be applied for this endorsement to be valid \\
  \midrule
  \multicolumn{2}{l}{\textbf{Variables of node $p_{i}$}} \\
  $\mathit{isSpeculative}_i$ & Whether $p_{i}$ speculatively applies transactions\\
  $\mathit{State}_{i}$  & Datastore state \\
  $T_{i}$      & Transactions endorsed by $p_{i}$ so far \\
  $E_{i,\id}$    & Endorsements received by $p_{i}$ for transaction $\id$ \\
  $\mathit{Policy}_{i}$ & Set of rules that define if $p_{i}$ agrees to apply given \\
               & transactions. This is not necessarily a deterministic \\
               & function and may involve human interaction \\
  \bottomrule
  \end{tabular}
\end{table}

As a result of leaderless quorums and conditional endorsements, transactions proceed through the life cycle presented in~\autoref{fig:tx_state}.
First, a client broadcasts a transaction to endorsers.
If it agrees with the transaction's operations, an endorser node can acknowledge the transaction by broadcasting its \emph{endorsement}.
If a threshold of valid endorsements is received within a transaction deadline (as defined in~\autoref{sec:protocol}), that transaction may enter the \textsc{Applicable} state. A transaction in that state has enough valid endorsements, but the node is not certain that those endorsements will remain valid - because of possible future conflicts.
The \textsc{Applied} state is an artifact introduced by the speculative execution of a transaction, when this mode is activated:
in this temporary state, the system cannot yet commit a transaction but it may execute the operations on the datastore state.
This optional optimization is useful to reduce global latency if the estimated probability of commit is very high.
Transactions can finally transition to final states \textsc{Committed}---once the node is sure that the endorsement will always stay valid---or \textsc{Dropped}, as we will detail in the following sections.

\section{The protocol}\label{sec:protocol}
The used variables and notations are summarized in~\autoref{table:notations}.

\subsection{Transaction applicability and endorsement validity}
\pgfdeclareshape{checkmark}{%
  \inheritsavedanchors[from=circle]
  \inheritanchorborder[from=circle]
  \inheritanchor[from=circle]{center}
  \backgroundpath{%
    \fill[scale=0.4](0,.20) -- (.20,0) -- (.5,.4) -- (.20,.12) -- cycle;
  }
}

\tikzstyle{q}=[fill=blue!20!white]
\tikzstyle{r}=[fill=yellow!50!white]

\begin{figure*}
  \begin{subfigure}[t]{\textwidth}
    \centering
    \begin{adjustbox}{max width=\textwidth,scale=1}
      \tikzstyle{bcastq}=[draw,circle,fill,minimum size=1mm,inner sep=0pt]
\tikzstyle{bcastr}=[draw,circle,fill=white,minimum size=1mm,inner sep=0pt]
\tikzstyle{endorsement}=[draw,fill=white,semithick,circle]
\tikzstyle{transaction}=[draw,semithick,rectangle,minimum size=3.5mm]
\tikzstyle{q}=[fill=blue!20!white]
\tikzstyle{r}=[fill=yellow!50!white]
\tikzstyle{msg}=[->,>=stealth,semithick]
\tikzstyle{legend}=[anchor=west,font=\scriptsize]

\begin{tikzpicture}[x=1cm,y=0.5cm]
  \def\XMAX{11}
  \def\XFRONT{10.5}
  \pgfmathsetmacro\XLEGEND{\XMAX+0.8};
  \pgfmathsetmacro\XLEGENDTEXT{\XLEGEND+0.3};

  \node[draw=none] (p1) at (0, -1) {$p_{1}$};
  \node[draw=none] (p2) at (0, -2) {$p_{2}$};
  \node[draw=none] (p3) at (0, -3) {$p_{3}$};
  \node[draw=none] (p4) at (0, -4) {$p_{4}$};

  \draw[->] (p1) -- +(\XMAX, 0);
  \draw[->] (p2) -- +(\XMAX, 0);
  \draw[->] (p3) -- +(\XMAX, 0);
  \draw[->] (p4) -- +(\XMAX, 0);

  \node[draw=none] at (1, 0) {$q$};
  \node[transaction,q] at (1, -1) {};
  \node[bcastq] (q) at (1, -1) {};
  \draw[msg] (q) .. controls +(0.2, 0.2) .. +(0.4, 0);
  \draw[msg] (q) -- (1.5, -2);
  \draw[msg] (q) -- (1.4, -3);
  \draw[msg] (q) -- (1.8, -4);

  \node[draw=none] at (3, 0) {$e_{q,i}$};
  \node[endorsement,q] at (2.7, -1) {};
  \node[bcastq] (eq1) at (2.7, -1) {};
  \node[endorsement,q] at (3, -2) {};
  \node[bcastq] (eq2) at (3, -2) {};
  \draw[msg] (eq1) .. controls +(0.2, 0.2) .. +(0.4, 0);
  \draw[msg] (eq2) .. controls +(0.2, 0.2) .. +(0.4, 0);
  \draw[msg] (eq1) -- (3.7, -2) {};
  \draw[msg] (eq1) -- (3.2, -4) {};
  \draw[msg] (eq2) -- (3.8, -1) {};
  \draw[msg] (eq2) -- (3.4, -4) {};

  \node[draw=none] at (5, 0) {$r$};
  \node[transaction,r] at (5, -2) {};
  \node[bcastr] (r) at (5, -2) {};
  \draw[msg] (r) .. controls +(0.2, 0.2) .. +(0.4, 0);
  \draw[msg] (r) -- (5.5, -1);
  \draw[msg] (r) -- (6, -4);

  \node[draw=none] at (7, 0) {$e_{r,i}$};
  \node[endorsement,r] at (6.8, -1) {};
  \node[bcastr] (er1) at (6.8, -1) {};
  \node[endorsement,r] at (6.4, -2) {};
  \node[bcastr] (er2) at (6.4, -2) {};
  \node[endorsement,r] at (7, -4) {};
  \node[bcastr] (er4) at (7, -4) {};
  \draw[msg] (er1) .. controls +(0.2, 0.2) .. +(0.4, 0);
  \draw[msg] (er2) .. controls +(0.2, 0.2) .. +(0.4, 0);
  \draw[msg] (er4) .. controls +(0.2, 0.2) .. +(0.4, 0);
  \draw[msg] (er1) -- (7.3, -2) {};
  \draw[msg] (er1) -- (8, -4) {};
  \draw[msg] (er2) -- (7.5, -1) {};
  \draw[msg] (er2) -- (6.8, -4) {};
  \draw[msg] (er4) -- (7.8, -1) {};
  \draw[msg] (er4) -- (7.9, -2) {};

  \node[draw=none,rectangle,anchor=west,minimum width=60mm,color=red,fill=red!30,inner sep=3pt,rounded corners] at (2.2, -3) {};
  \node[draw=none] at (5, -3) {\scriptsize Node unavailable};
  \node[draw=none,star,star points=12,inner sep=3pt,red,fill=red!50] at (2.2, -3) {};

  \draw[draw,dashed,very thick] (\XFRONT, 0.5) -- ++(0, -5);
  \node[shape=circle,draw,inner sep=0.5,anchor=east,xshift=-5] at (\XFRONT, 0.5) {\scriptsize\ref{fig:before_eq3}};
  \node[shape=circle,draw,inner sep=0.5,anchor=west,xshift=5] at (\XFRONT, 0.5) {\scriptsize\ref{fig:after_eq3}};

  \node[draw=none] at (9, 0) {$e_{q,3}$};
  \node[endorsement,q] at (9, -3) {};
  \node[bcastq] (eq3) at (9, -3) {};
  \draw[msg] (eq3) .. controls +(0.2, 0.2) .. +(0.4, 0);
  \draw[msg] (eq3) -- (10.7, -1) {};
  \draw[msg] (eq3) -- (10.8, -2) {};
  \draw[msg] (eq3) -- (10.7, -4) {};

  \node[draw,anchor=west] at (1, -4.8) {\scriptsize $p_{4}$ is against $q$, ignoring it};

  \node[legend] at (12, 0) {\underline{Legend}};

  \node[transaction] at (\XLEGEND, -1.3) {};
  \node[legend] at (\XLEGENDTEXT, -1.3) {Transaction};

  \node[endorsement] at (\XLEGEND, -2.1) {};
  \node[legend] at (\XLEGENDTEXT, -2.1) {Endorsement};

  \node[transaction,q,draw=none] at (\XLEGEND, -2.9) {};
  \node[bcastq] at (\XLEGEND, -2.9) {};
  \node[legend] at (\XLEGENDTEXT, -2.9) {Message about $q$};

  \node[transaction,r,draw=none] at (\XLEGEND, -3.7) {};
  \node[bcastr] at (\XLEGEND, -3.7) {};
  \node[legend] at (\XLEGENDTEXT, -3.7) {Message about $r$};

\end{tikzpicture}
    \end{adjustbox}
    \caption{%
    Simplified history of messages exchanged between the four nodes in our example.
    Node $p_{1}$ first submits transaction $q$, that is endorsed by both $p_{1}$ and $p_2$ through $e_{q,1}$ and $e_{q,2}$.
    Shortly after, $p_{2}$ submits a conflicting transaction $r$ that is endorsed by 3 nodes:
    with conditions for $p_{1}$ and $p_{2}$ (see~\ref{fig:before_eq3}) and without condition for $p_{4}$.
    After a period of unavailability, $p_{3}$ broadcasts its endorsement of $q$, leading to state~\ref{fig:after_eq3}.
}\label{fig:messages}
  \end{subfigure}

  \vspace{3mm}

  \begin{subfigure}[t]{.45\textwidth}
    \centering

\begin{tikzpicture}[
  node distance=1mm,
  query/.style={rectangle,minimum size=4mm,draw=black!80,thick,xshift=5mm},
  endorsement/.style={circle,inner sep=0.5pt,draw=black,thick},
  arrow/.style={->,>=stealth,thick,out=0,in=180},
  check/.style={checkmark,draw=green,xshift=0.5mm,yshift=0.5mm},
  ]

  \node (eq1) [endorsement,q] {$e_{q,1}$};
  \node (eq2) [endorsement,q,below=of eq1] {$e_{q,2}$};
  \node (q)   [query,q,right=of eq2] {$q$};
  \begin{scope}[node distance=-0.7mm]
    \node (i)   [right=of q,xshift=15mm] {};
    \node (er1) [endorsement,r,above=of i] {$e_{r,1}$};
    \node (er2) [endorsement,r,below=of i] {$e_{r,2}$};
  \end{scope}
  \node (er4) [endorsement,r,below=of er2] {$e_{r,4}$};
  \node (r)   [query,r,right=of er2] {$r$} node at (r) [check] {};

  \draw [arrow] (eq1) to (q);
  \draw [arrow] (eq2) to (q);
  \draw [arrow,densely dashed] (q) to (er1);
  \draw [arrow,densely dashed] (q) to (er2);
  \draw [arrow] (er1) to (r);
  \draw [arrow] (er2) to (r);
  \draw [arrow] (er4) to (r);
\end{tikzpicture}

  \captionsetup{justification=centering}
  \caption{$r$ is \textsc{Applicable} given its 3 valid endorsements.}\label{fig:before_eq3}
  \end{subfigure}
  \begin{subfigure}[t]{.05\textwidth}
    \centering
\begin{tikzpicture}
  \draw[draw, dashed, very thick] (0, 0) -- ++(0, 2.7);
  \node[shape=circle,draw,inner sep=0.5,anchor=east,xshift=-5] at (0,2.5) {\scriptsize\ref{fig:before_eq3}};
  \node[shape=circle,draw,inner sep=0.5,anchor=west,xshift=5]  at (0,2.5) {\scriptsize\ref{fig:after_eq3}};
\end{tikzpicture}
  \end{subfigure}
  \begin{subfigure}[t]{.45\textwidth}
    \centering

\begin{tikzpicture}[
  node distance=1mm,
  query/.style={rectangle,minimum size=4mm,draw=black!80,thick,xshift=5mm},
  endorsement/.style={circle,inner sep=0.5pt,draw=black,thick},
  arrow/.style={->,>=stealth,thick,out=0,in=180},
  cross/.style={cross out,draw=red,line width=1mm,line cap=round,xshift=3mm,yshift=3mm},
  check/.style={checkmark,draw=green,xshift=0.5mm,yshift=0.5mm},
  red/.style={draw=red,ultra thick},
  ]

  \node (eq1) [endorsement,q] {$e_{q,1}$};
  \node (eq2) [endorsement,q,below=of eq1] {$e_{q,2}$};
  \node (eq3) [endorsement,q,below=of eq2,ultra thick] {$e_{q,3}$};
  \node (q)   [query,q,right=of eq2,ultra thick] {$q$} node at (q) [check] {};
  \begin{scope}[node distance=-0.7mm]
    \node (i)   [right=of q,xshift=15mm] {};
    \node (er1) [endorsement,r,above=of i,red] {$e_{r,1}$} node at (er1) [cross] {};
    \node (er2) [endorsement,r,below=of i,red] {$e_{r,2}$} node at (er2) [cross] {};
  \end{scope}
  \node (er4) [endorsement,r,below=of er2] {$e_{r,4}$};
  \node (r)   [query,r,right=of er2] {$r$};

  \draw [arrow] (eq1) to (q);
  \draw [arrow] (eq2) to (q);
  \draw [arrow,ultra thick] (eq3) to (q);
  \draw [arrow,ultra thick,densely dashed] (q) to (er1);
  \draw [arrow,ultra thick,densely dashed] (q) to (er2);
  \draw [arrow] (er1) to (r);
  \draw [arrow] (er2) to (r);
  \draw [arrow] (er4) to (r);
\end{tikzpicture}

  \captionsetup{justification=centering}
  \caption{Since $q$ is now \textsc{Applicable}, $e_{r,1}$ and $e_{r,2}$ are no longer valid. Hence, $r$ is no longer \textsc{Applicable}.}\label{fig:after_eq3}
  \end{subfigure}

  \caption{%
    Example of graph of conditions for transactions $q, r$ and their respective endorsements $e_{q,i}$ and $e_{r,i}$. $e_{r,1}$ and $e_{r,2}$ are conditioned by $q$,
    while other endorsements are not. We set $\omega = 3$.
    (\ref{fig:before_eq3}) shows the knowledge of correct nodes before the arrival of $e_{q,3}$ (\ref{fig:after_eq3}).
  }\label{fig:definitions}
\end{figure*}
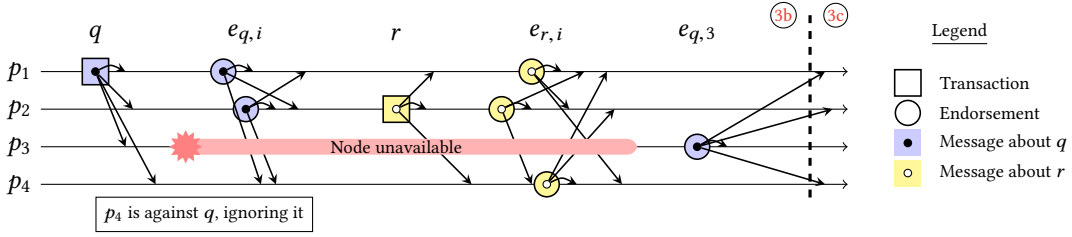
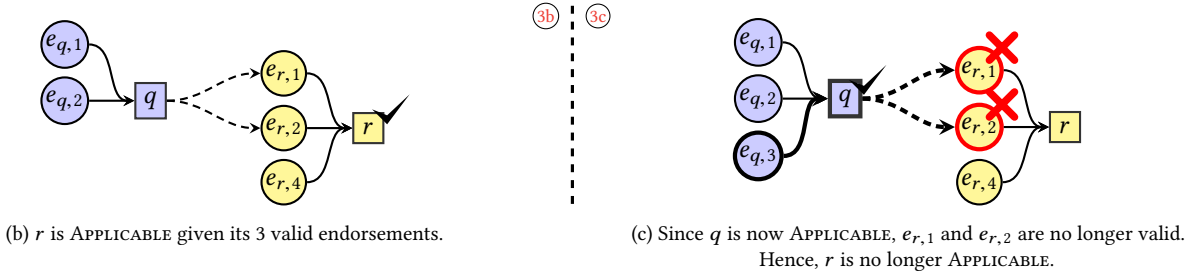

The notion of \emph{applicable transactions} (\autoref{fig:tx_state}) plays a key role in the eventual consistency of \proj{}, and is recursively defined in terms of \emph{valid endorsements}.
More precisely:
\begin{itemize}
 \item A transaction $t$ is \textsc{Applicable} at node $p_{i}$ if and only if there exists at least $\omega$ \textsc{Valid} endorsements for $t$ at node $p_{i}$, where $\omega$ is a Byzantine quorum threshold, chosen to be larger than $\lfloor \frac{n+f}{2} \rfloor$.
 \item An endorsement $e=\langle \id, i, C \rangle$ of a transaction $t=\langle \id, d, R, \Delta \rangle$ with ($e.\id = t.\id$) is \textsc{Valid} at node $p_{i}$ if and only if every transaction $c$ in the condition set $e.C$ of $e$ has an earlier deadline than $t$ and is not \textsc{Applicable}.
   A transaction deadline is set by its issuer\ and constrained by system-wide policies to avoid excessively-large deadlines.
\end{itemize}

The interplay between these two notions drives how a transaction proceeds through the state diagram of \autoref{fig:tx_state}, and is illustrated on the scenario shown in \autoref{fig:definitions}.
In this example, Nodes $p_1$ and $p_2$ propose two conflicting transactions $q$ and $r$ (\autoref{fig:messages}).
$q$ is at first only endorsed by $p_1$ and $p_2$. ($e_{x,i}$ denotes the endorsement of transaction $x$ by node $p_i$.)
When transaction $r$ is broadcast, $p_1$ and $p_2$ detect a potential conflict with $q$, which they have already endorsed, and issue \emph{conditional} endorsements for $r$. $p_4$ has not endorsed $q$: it can endorse $r$ unconditionally.

The resulting condition graph on every node at this point is shown in \autoref{fig:before_eq3}. Endorsement conditions are represented by dashed lines: for instance, $e_{r,1}$ is valid if $q$ is not \textsc{Applicable}.
In \autoref{fig:before_eq3}, $q$ has only received 2 endorsements, and is therefore not applicable under a quorum threshold of $\omega=3$. $r$ has received 3 endorsements (from $p_{1,2,4}$), all of which are valid: $e_{r,4}$ because its condition set is empty, $e_{r,1}$ and $e_{r,2}$ because $q$ is not applicable.
Transaction $r$ is therefore \textsc{Applicable}, and may be speculatively executed but cannot be \textsc{Committed} yet as $q$ has not been \textsc{Dropped}.

When a third endorsement $e_{q,3}$ for $q$ is finally received from $p_3$, the condition graph of each node changes to that of \autoref{fig:after_eq3}.
At this point, the minimum number of valid endorsements is now reached for $q$, making two endorsements for $r$ invalid.
$q$ is now \textsc{Applicable} while $r$ is no longer so.

\subsection{Algorithm}

The detail of \proj{}'s workings is presented in Algorithms~\ref{callalgo}, \ref{localalgo}, \ref{applicablealgo} and~\ref{checkstatealgo}.
Our design is reactive: endorsers and observers react to the \textsc{Transaction} and \textsc{Endorsement} messages they receive from the network.
For simplicity, we do not include authentication and invariant checks. (In the following, `line~x.y'\ refers to line~y of Algorithm~x.)

\begin{algorithm}[tb]
  \caption{Message callbacks at node $p_{i}$}\label{callalgo}
\begin{algorithmic}[1]
  \Upon[Transaction]{$\id, d, R, \Delta$}\label{line:receivingTrans}
    \State$t \gets \textsc{Transaction}\langle \id, d, R, \Delta \rangle$
    \State$done \gets \bot$
    \InlineComment{Continue until no active conflicting transaction present}
    \While{$done = \bot$}\label{retryendorse}
    \If{\Call{CanEndorse}{$t$}}\label{line:invoke-canendorse}
      \State$C \gets \{ c : c \in T_{i} \wedge \Gamma(c.\Delta, \Delta) \}$\label{conflictingtransactions}
      \If{$C = \emptyset$}\label{line:conflictingtransactions-empty}
        \InlineComment{No conflicting transaction}
        \State\Call{Endorse}{$t, \emptyset$}\label{unconditional}
        \State$done \gets \top$
        \Else
          \If{$\forall c \in C, c.d \leq now()$}\label{conditionaltrigger}
            \InlineComment{Expired conflicting transactions}
            \State\Call{Endorse}{$t, C$}\label{conditional}
            \State$done \gets \top$
          \EndIf
          \InlineComment{Otherwise, not done, going back to start of while loop}  
      \EndIf
    \Else
      \InlineComment{Unable to endorse}
      \State$done \gets \top$
    \EndIf
    \EndWhile\label{retryendorse:end}
  \EndUpon

  \Statex{}
  \Upon[Endorsement]{$\id, j, C$}\label{line:new-endorsement}
    \State$E_{i,\id} \gets E_{i,\id} \cup \{\textsc{Endorsement}\langle \id, j, C \rangle\}$
    \State$\forall t \in T_{i} : \Call{CheckState}{t}$
  \EndUpon
\end{algorithmic}
\end{algorithm}

\begin{algorithm}[tb]
\caption{Endorsement checks at node $p_{i}$}\label{localalgo}
\begin{algorithmic}[1]
  \Function{CanEndorse}{$t$}
    \If{$t.d \leq now()$}\label{nowcondition}
      \State\textbf{return} \textbf{abort}\Comment{Timeout}
    \EndIf

    \If{$\mathit{State_{i}}$ not compatible with $t.R$}
        \State\textbf{return} \textbf{abort}\Comment{Consistency}
    \EndIf

    \State$\mathit{State}' \gets t.\Delta(\mathit{State})$
    \If{$\mathit{State}'$ does not comply to $\mathit{Policy}_{i}$}
      \State\textbf{return} \textbf{abort}\Comment{Policy}\label{localpolicy}
    \EndIf

    \State\textbf{return} \textbf{OK}
  \EndFunction

  \Statex{}
  \Function{Endorse}{$t, C$}
    \State\Call{Broadcast}{\textsc{Endorsement}$\langle t.\id, i, C \rangle$}
    \State$T_{i} \gets T_{i} \cup t.\id$
  \EndFunction
\end{algorithmic}
\end{algorithm}

\begin{algorithm}[tb]
  \caption{Predicates at node $p_{i}$}\label{applicablealgo}
\begin{algorithmic}[1]
  \Function{Applicable}{$\id$}
    \State$E^{+}_{i,\id} \gets \{ e : e \in E_{i,\id} \wedge \Call{Valid}{$e$} \}$\label{recursion}
    \State\textbf{return} $|E^{+}_{i,\id}| \geq \omega$\label{quorum}
  \EndFunction

  \Function{Valid}{$e$}
    \State\textbf{return} $\forall c \in e.C, \neg \Call{Applicable}{c.\id}$\label{invalidendorsement}
  \EndFunction
\end{algorithmic}
\end{algorithm}

A client starts a set of operations by broadcasting a \textsc{Transaction}$\langle \id, d, R, \Delta \rangle$ to nodes, with a configurable deadline $d$ and a set of operations $\Delta$.
On receiving this $\textsc{Transaction}$ (line~\algref{callalgo}{line:receivingTrans}), each endorser first checks whether the transaction can be endorsed ({\sc CanEndorse}() at line~\algref{callalgo}{line:invoke-canendorse}, described in Algorithm~\ref{localalgo}).
In particular, endorsers must check that the transaction's deadline has not been reached with respect to their local clock (line~\algref{localalgo}{nowcondition}).
Endorsers can also deliberately choose \textbf{not} to endorse a transaction simply by ignoring it, for local policy reasons (line~\algref{localalgo}{localpolicy}).
If {\sc CanEndorse}() returns \emph{true}, each endorser $p_{i}$ then checks that it has not already endorsed conflicting transactions $C$  (lines~\algref{callalgo}{conflictingtransactions}-\algref{callalgo}{line:conflictingtransactions-empty}).
The predicate $\Gamma$ returns \emph{true} if the two transactions passed as arguments are in conflict.
Three cases may happen:
\begin{itemize}
  \item If no conflicting transaction  exists, $p_{e}$ can broadcast its $\textsc{Endorsement}\langle \id, i, \emptyset \rangle$ \textit{without condition} (line~\algref{callalgo}{unconditional}).
  \item If $C$ only contains outdated transactions, $p_{e}$ can broadcast a conditional $\textsc{Endorsement}\langle \id, i, C \rangle$, allowing the application of the transaction given the non applicability of every outdated transactions (line~\algref{callalgo}{conditional}).
  \item Otherwise, $p_{e}$ must wait until conflicting deadlines are over, and restarts the while loop (line~\algref{callalgo}{retryendorse}).
\end{itemize}

New endorsements are received at line~\algref{callalgo}{line:new-endorsement}, and trigger the execution of the {\sc CheckState}() function (described in Algorithm~\ref{checkstatealgo}) on all transactions already endorsed by the receiving endorser ($T_i$ set).
{\sc CheckState}() ensures that the state of the datastore $\mathit{State}_{i}$ is consistent with the \textsc{Applicable} state of transactions (lines~\algref{checkstatealgo}{rollback} and~\algref{checkstatealgo}{apply}).
It also triggers the \textsc{Commit} operation on transactions $q$ when there are a sufficient number of unconditional endorsements on $t$ (line~\algref{checkstatealgo}{commit}).
Finally, the procedure can decide to trigger checkpoints when conditions are blocking newer transactions (represented by the \textsc{Old} trigger, tested at line~\algref{checkstatealgo}{oldcall}).

More specifically, once a node has received a predefined quorum $\omega$ of valid and distinct endorsements for a given transaction $t$ (implemented by the {\sc Applicable}() and {\sc Valid}() functions in Algorithm~\ref{applicablealgo}, invoked at line~\algref{checkstatealgo}{line:calling-applicable}), {\sc CheckState}() applies $t.\Delta$ if the node is configured to execute applicable transactions speculatively (line~\algref{checkstatealgo}{apply}).
Coming back to~\autoref{fig:tx_state}, it means that the transaction moves either to the \textsc{Applicable} state (if the node is not speculative) or the \textsc{Applied} state otherwise.

We must ensure that $\omega > \left \lfloor \frac{n+f}{2} \right \rfloor$ to tackle Byzantine endorsements.
Higher $\omega$ values allow to build stricter transaction acceptance rules, requiring up to unanimous agreement ($\omega = n$), but this comes at the cost of availability by depending on Byzantine nodes to endorse transactions.
(The minimum number of nodes to allow both availability and safety is $n \geq 3f + 1$~\cite{Malkhi1998a}.)

The \textsc{CheckState}() function is also used to verify the validity of previously-valid endorsements because of \textit{endorsement conditions} (line~\algref{applicablealgo}{invalidendorsement}),
potentially triggering transaction rollback(s) (line~\algref{line:calling-applicable}{rollback}, as illustrated in~\autoref{fig:definitions}).
A transaction can move back and forth from its initial \textsc{Pending} state to the \textsc{Applicable} state.
It is important to note that those states are \emph{local}: each node may have a different view of \textsc{Applicable} transactions depending on the messages it has received.
However, our safety property guarantees that no transaction can both be \textsc{Committed} at a correct node $p_{i}$ and \textsc{Dropped} at another correct node $p_{j}$.
Conversely, if a transaction reaches one of those two final states at a correct node $p_{i}$, every other correct node will eventually set the same state for that transaction.
We revisit these points in \autoref{sec:proofs}, where we formally prove some of \proj{}'s key properties.

\subsection{Checkpointing}

In many cases, we expect that a node can conclude from received endorsements that the \textsc{Applicable} predicate has reached a final state (true or false) by analyzing the transaction's graph of conditions.
When complex dependencies arise between endorsements and transactions, some transactions might however interlock.
As an example in~\autoref{fig:definitions}, nodes cannot know whether $r$ must be committed before receiving $e_{q,3}$.
To cope with this issue and ensure both liveness and consistency, we use a simple checkpoint sub-protocol (Algorithm~\ref{checkpointalgo}) to prune the condition graph and unblock transactions.
This sub-protocol builds upon an underlying \emph{Byzantine Veto Procedure} (BVP) in which each node $p_{i}$ proposes a choice $c_{i} \in \{0, 1\}$ and decides a final value $d_{i}$.
BVP is a Byzantine-tolerant version of the Non-Blocking Atomic Commitment (NBAC) protocol~\cite{Babaoglu1993}, and is expected to satisfy the following properties \emph{with eventually-synchronous communications}:
1) \textbf{\emph{Termination}}: every correct node eventually decides on a value;
2) \textbf{\emph{Agreement}}: no two correct nodes decide on different values; and
3) \textbf{\emph{Validity}}: if a correct node decides 1, then all correct nodes proposed 1 (equivalently, if any correct node proposes 0, then a correct node decides 0).
We return to the implementation details of BVP in section~\ref{sec:implementation}.

\begin{algorithm}[tb]
  \caption{State checking at node $p_{i}$}\label{checkstatealgo}
\begin{algorithmic}[1]
  \Function{CheckState}{$t$}\label{func:checkstate}
    \If{$\neg\textsc{Applicable}(t)$}\label{line:calling-applicable}
      \If{$\textsc{Applied}(t)$}
        \State$\mathit{State}_{i} \gets \Call{Rollback}{\mathit{State}_{i}, t}$\label{rollback}
      \EndIf
    \Else
      \If{$\neg\textsc{Applied}(t)$ \textbf{and} $\mathit{isSpeculative}_i$}
        \State$\mathit{State}_{i} \gets \Call{Apply}{\mathit{State}_{i}, t}$\label{apply}\Comment{Speculative execution}
      \EndIf

      \InlineComment{Endorsements that will always stay valid}
      \State$\Sigma_{i,t} = \{ e \in E_{i,t} : e.C = \emptyset \}$\label{commit}
      \If{$|\Sigma_{i,t}| \geq \omega$}
        \State\textbf{if} $\neg\textsc{Applied}(t)$ \textbf{then}
          $\mathit{State}_{i} \gets \Call{Apply}{\mathit{State}_{i}, t}$
        \State$\mathit{State}_{i} \gets \Call{Commit}{\mathit{State}_{i}, t}$
        \State$T_{i} \gets T_{i} \setminus \{t\}$
      \EndIf

      \InlineComment{Conditions that could be dropped}
      \State$\bar{T} \gets \{ \forall e \in E_{i, t}, c \in e.C : \textsc{Old}(c) \}$\label{oldcall}
      \If{$|\bar{T}| \geq 1$}
        \State\Call{StartCheckpoint}{$\bar{T}$}
      \EndIf
    \EndIf
  \EndFunction

  \Statex{}
  \InlineComment{Example of checkpoint trigger for configurable delay $\delta$}
  \Function{Old}{$t$}
    \State\textbf{return} $\neg\textsc{Applicable}(t) \wedge t.d < (now() - \delta)$
  \EndFunction
\end{algorithmic}
\end{algorithm}

\begin{algorithm}[tb]
  \caption{Checkpoint at node $p_{i}$}\label{checkpointalgo}
\begin{algorithmic}[1]
  \Function{StartCheckpoint}{$\bar{T}$}
    \State\Call{Broadcast}{\textsc{Checkpoint}$\langle \bar{T} \rangle$}\label{broadcastcheckpoint}
  \EndFunction

  \Statex{}
  \Upon[Checkpoint]{$\bar{T}$}\label{line:msg:checkpoint}
    \State$c \gets \left\{\begin{matrix*}[l]
      0~\emph{if}~\exists t \in \bar{T} :
      \textsc{Applicable}(t) \vee \textsc{Committed}(t)
      \\
      1~\emph{otherwise}
    \end{matrix*}\right.$\label{choice}
    \State$\mathit{decision} \gets \Call{BVP}{\bar{T}, c}$\label{line:startBVP}

    \Statex{}
    \If{$\mathit{decision} = 1$}\Comment{Cleanup}
      \State$T_{i} \gets T_{i} \setminus \bar{T}$\label{pruning-start}
        \Comment{Drop transactions}
      \State$\forall t \in \bar{T} : E_{i, t} = \emptyset$\label{forgetendorsements}
        \Comment{Forget endorsements of dropped transactions}   
      \State$\forall t \in T_{i}, e \in E_{i, t} : e.C = e.C \setminus \bar{T}$\label{pruning-end}\label{forgetconditions}
        \Comment{Forget conditions}
      \State$\forall t \in T_{i} \cup \bar{T} : \Call{CheckState}{t}$
    \EndIf
  \EndUpon
\end{algorithmic}
\end{algorithm}

When a node decides to start a checkpoint, it triggers a BVP instance with a \textsc{Checkpoint} proposal (line~\algref{checkpointalgo}{line:startBVP}), a set of transactions representing a cut of their graph of conditions.
Each proposal aims at removing old transactions that block newer transactions from being committed.
Informally, a proposal might be as simple as \textit{``transaction $t$ will never be applicable, drop it''}.
During the procedure, correct nodes are expected to propose $0$ (``Veto'') if and only if they hold evidence that the checkpoint proposal is wrong (line~\algref{checkpointalgo}{choice}).
(Such nodes must submit this evidence in the form of signed endorsements.)
Two checkpoint results are possible per invocation:
(1) If the final decision is $1$, correct nodes can prune their local graph of conditions according to the confirmed proposal (lines~\algref{checkpointalgo}{pruning-start}-\algref{checkpointalgo}{pruning-end});
(2) otherwise, some correct nodes have reasons for blocking the checkpoint proposal.
After having added the evidence(s) to their graph of conditions, correct nodes can discard this checkpoint instance.

In our example from figure~\ref{fig:before_eq3}, if the BVP decision on the proposal \textit{``drop $q$''} is $1$,
then every node can confidently drop $q$ and remove $q$'s condition on the endorsements $e_{r,i}$, thus effectively commiting $r$.
On the contrary, if the BVP decision is 0, correct nodes can expect an evidence going against the proposal:
for instance, node $p_{3}$ can broadcast $e_{q,3}$ again.
This allows nodes to progress, finally triggering the commit of $q$ and the drop of $r$ for every node.
We discuss and evaluate the overhead of this checkpoint procedure in~\autoref{sec:checkpoints_eval}.

\subsection{Eventual consistency: proofs}\label{sec:proofs}

We first show that every correct node eventually obtains the same set of applicable transactions.
We then show that transactions entering the final committed and dropped states will stay in these states for every correct node.

\begin{lemma}[Acyclic conditions]\label{lem:acyclic}
  Let $q, r$ be two transactions with $q.d \leq r.d$.
  There cannot be any \textsc{Endorsement}$\langle q.\id, i, C \rangle$ broadcasted by a correct node $p_{i}$ with $r \in C$.
\end{lemma}

\begin{proof}
In our algorithm, the only case where conditional endorsements are broadcasted is at line~\algref{callalgo}{conditional}.
For this line to be executed, \textsc{CanEndorse}$(q)$ must have returned \textbf{OK}.
Hence, per line~\algref{localalgo}{nowcondition}, we must have $q.d > now()$.
Given line~\algref{callalgo}{conditionaltrigger}, every element $r$ of $C$ must fulfill $r.d \leq now()$.
If we assume that local operations execute instantaneously, the value of $p_{i}$'s local clock $now$ shall be the same in the two constraints.
We have $\forall r \in C, r.d < q.d$.
\end{proof}

We can use the result of this lemma to filter incoming endorsements at each node, and detect Byzantine behavior.
In the following, we suppose that every malformed endorsement has correctly been filtered by correct nodes.

\begin{proposition}[Termination]\label{lem:termination}
  Assuming the BVP protocol terminates, every proposed functions and callbacks terminate.
\end{proposition}

\begin{proof}
This is trivial for functions \textsc{CanEndorse} and \textsc{Endorse} in Algorithm~\ref{localalgo}.
When receiving a \textsc{Transaction} message at line~\algref{callalgo}{line:receivingTrans}, a node may execute the while loop (lines~\algref{callalgo}{retryendorse}-\algref{callalgo}{retryendorse:end}) several times if conflicts are detected.
A node is, however, guaranteed to exit the while loop when \textsc{CanEndorse} returns false because the transaction's deadline is over (timeout clause).
Upon reception of the messages \textsc{Endorsement} (line~\algref{callalgo}{line:new-endorsement}), and \textsc{Checkpoint} (line~\algref{checkstatealgo}{line:msg:checkpoint}), and for the function \textsc{CheckState} (line~\algref{checkstatealgo}{func:checkstate}), the termination is conditioned by the termination of the Binary Veto Procedure (BVP) and the \textsc{Applicable} predicate.

BVP terminates by assumption.
The case of \textsc{Applicable} is somewhat more involved, as the functions \textsc{Applicable} and \textsc{Valid} recursively call each other.
Every \textsc{Applicable} call of transaction $r$ will trigger a finite number of \textsc{Valid} calls on endorsements $e \in E_{i,r.\id}$ received for $r$ (line~\algref{applicablealgo}{recursion}).
Each \textsc{Valid} call will in turn call a finite number of \textsc{Applicable} call for every $c \in e.C$ (line~\algref{applicablealgo}{invalidendorsement}).
Because of Lemma~\ref{lem:acyclic}, we know that for two transactions $q, r$ with $q.d \leq r.d$, there can be no \textsc{Endorsement}$\langle q.\id, i, C \rangle$ with $r \in C$ in any $E_{j,q.\id}$ set.
This property implies that the recursion between the two predicates will call \textsc{Applicable} with transactions ordered by decreasing deadline $q.d$, thus eliminating any loop.
\end{proof}

\begin{lemma}[Local safety]\label{lem:local}
  Let $p_{i}$ be a correct node and $q, r$ two transactions with $q \neq r$:
  $$\Gamma(q.\Delta, r.\Delta) \Rightarrow \neg (\textsc{Applicable}(q)~\wedge~\textsc{Applicable}(r))$$
\end{lemma}

\begin{proof}
Without loss of generality, we suppose $q.d \leq r.d$ and that \textsc{Applicable}(q) and \textsc{Applicable}(r) hold at correct node $p_{i}$.
Let us note $C_{\emptyset}$ a condition set such that $q \notin C_{\emptyset}$.
Given Lemma~\ref{lem:acyclic} and instruction~\algref{applicablealgo}{invalidendorsement}, the only endorsements for $q$ and $r$ that are valid for $p_{i}$ are
$E^{+}_{i,q} = \{\textsc{Endorsement}\-\langle q, j, C_{\emptyset} \rangle\}$
and
$E^{+}_{i,r} = \{\textsc{Endorsement}\langle r, j, C_{\emptyset} \rangle\}$
, for any endorser $p_{j}$.
According to line~\algref{applicablealgo}{quorum}, we must have $\omega \leq |E^{+}_{i,q}|$.
Since correct nodes cannot send both kind of endorsements per algorithm~\autoref{callalgo}, we also must have $\omega \leq |E^{+}_{i,r}| \leq n - |E^{+}_{i,q}| + f \leq n - \omega + f$, or more simply $\omega \leq \frac{n+f}{2}$.
This leads to a contradiction with $\left \lfloor \frac{n + f}{2}  \right \rfloor + 1 \leq \omega$.
\end{proof}

\begin{lemma}[Eventual consistency]\label{lem:global}
  After enough time, for two correct nodes $\{i, j\}$, if \textsc{Applicable}(q) holds at node $p_{i}$, then \textsc{Applicable}(q) must hold at node $p_{j}$.
\end{lemma}

\begin{proof}
Given the reliable broadcast and the eventual synchrony assumptions, every correct node will receive the same set of endorsements.
Since we have $E_{i,q.\id} = E_{j,q.\id}$, the value of \textsc{Applicable}(q) must be the same for $p_{i}$ and $p_{j}$ nodes before checkpointing.
During every checkpoint, endorsement sets and conditions may be modified by lines~\algref{checkpointalgo}{forgetendorsements} and~\algref{checkpointalgo}{forgetconditions}.
Given the agreement property of the BVP, every correct node will prune their graph of conditions according to the confirmed proposal, or no node will do.
\end{proof}

\begin{proposition}[Durability]\label{prop:durability}
  The proposed algorithm ensures that if a transaction $t$ is \textsc{Committed} (resp. \textsc{Dropped}) at a correct node $p_{i}$, it will stay in this state and will eventually be \textsc{Committed} (resp. \textsc{Dropped}) for every other correct node.
\end{proposition}

\begin{proof}
The \textsc{Committed} state is triggered in the \textsc{CheckState} routine when a transaction obtains a number of unconditional endorsements $k \geq \omega$ (line~\algref{checkstatealgo}{commit}), either after receiving a new endorsement or after a successful checkpoint.
With the same arguments than Lemma~\ref{lem:global}, we know that in the first case every correct node will \textsc{Commit} $t$.
The second case is covered by the underlying BVP during a checkpoint: after a successful checkpoint ($decision = 1$), every correct node prunes its conditions graph in the same way, thus eventually triggering \textsc{Commit} operations on every correct node (resp. \textsc{Drop}, line~\algref{checkpointalgo}{pruning-start}).
Per line~\algref{checkstatealgo}{rollback}, operating a \textsc{Rollback} on a transaction $t$ is only possible if \textsc{Applicable}($t$) does not hold anymore.
Since we know that at least $\omega$ unconditional endorsements for $t$ have been received at this point, the only way that the predicate would not hold is due to a successful checkpoint on $t \in \bar{T}$, with the ``endorsement-forgetting'' operation depicted at line~\algref{checkpointalgo}{forgetendorsements}.
However, given line~\algref{checkpointalgo}{choice}, if at least one correct node has \textsc{Committed} $t$, no further checkpoint on $t \in \bar{T}$ can return a decision of 1.
No \textsc{Dropped} transaction could become \textsc{Applicable} again in correct nodes due to the pruning of endorsements after a successful checkpoint.
\end{proof}

It follows from Lemmas~\ref{lem:local} and~\ref{lem:global} that no two conflicting operations can be committed in different orders by two correct nodes, thus ensuring \emph{eventual consistency} with Proposition~\ref{prop:durability}.
Another core feature of our algorithm is
the ability for correct nodes to \textit{reject any transaction without giving their reasons}, as underlined in line~\algref{localalgo}{localpolicy}.
We formally define this property as the system's \textit{fairness}.

\begin{proposition}[Fairness]\label{theo:fairness}
  If no majority of correct nodes $k > \left \lfloor \frac{n-f}{2} \right \rfloor$ endorsed a transaction $t$, \textsc{Applicable}($t$) will never hold at any correct node.
\end{proposition}

\begin{proof}
Let $\epsilon$ be the number of endorsements for $t$.
We must have $\omega \leq \epsilon \leq \left \lfloor \frac{n-f}{2} \right \rfloor + f$ for $q$ to be \textsc{Applicable}, which is impossible due to our definition of $\omega$.
\end{proof}

\newpage{}
\section{Implementation}\label{sec:implementation}

We have implemented \proj{} in Go~\cite{TheGoAuthors,PnyxDBSourceCode}.
In this section, we describe our implementation of node authentication using a web of trust,
along with practical solutions for the assumed algorithm primitives, namely the \textit{Reliable Broadcast} and the \textit{Byzantine Veto Procedure} (BVP).

Our technical choices were driven by the common size of consortia and state of the art cryptography techniques.
Hence, we decided to target a scale of several hundreds to thousands of nodes per network, excluding clients.

\subsection{Web of trust and policy files}\label{sec:wot}

Nodes need to be authorized to participate in a closed consortium.
The Hyperledger Fabric consortium blockchain proposes a centralized approach, where a single authority gives cryptographic certificates to network members~\cite{Androulaki2018}.
However, a corrupted authority may introduce a large number of malicious nodes in the network, potentially breaking the assumption on the maximum number of faulty nodes $f$.

Our implementation relies instead on a web of trust and policy files, inspired from PGP~\cite{Abdul-Rahman1997}.
The web of trust is used to link nodes' identities with their public key, providing a sound authentication mechanism.
Our implementation supports several cryptographic authentication schemes, and uses the recognized fast ed25519 procedure by default~\cite{Bernstein}.

Nodes need to know the identities of \emph{endorsers}, along with useful metadata such as authorized operations and default network parameters.
We use a \emph{universal policy file} for this, and we expect nodes to agree on the content of this policy file: this is similar to the distribution of a common \textit{genesis} file required by a number of existing BFT systems~\cite{Bessani2014,Buchman2016,Androulaki2018}.
Classic \proj{} transactions could be used to update the universal policy in a consistent and democratic way, for instance as done in the Tezos Blockchain~\cite{Goodman2014}.

\subsection{Reliable broadcast and recovery}\label{sec:gossip}

A Byzantine-resilient reliable broadcast is required in \proj{} to ensure that correct nodes will eventually receive every transaction and endorsement, possibly out-of-order.
Such an algorithm was proposed by Bracha~\cite{Bracha1987}, but it has a message complexity of $O(n^{2})$, which makes it impractical for our targeted scale.
Based on current public and consortium blockchains implementations~\cite{Bano}, we propose a probabilistic gossip algorithm as our reliable broadcast primitive,
where each node communicates only with a small number of neighbors to lower the total message complexity.
Such algorithms are known to disseminate information with a logarithmic number of messages and are used in popular BFT public and consortium blockchain networks.
We selected GossipSub~\cite{JeromyJohnson} from the libp2p project as our gossip broadcast algorithm.
Libp2p is a popular set of libraries for peer-to-peer communication, that targets gossip networks of 10000 nodes with practical Byzantine Fault Tolerance.
It comes with standard mechanisms for inter-node communication and authentication that fulfill our specifications, and supports our default ed25519 authentication scheme.

Using a gossip algorithm as our broadcast primitive inherently introduces uncertainty in the reliability of the broadcast~\cite{Guerraoui2019}.
We propose to complement this probabilistic broadcast with \textit{retransmissions} and \textit{state transfers}: with very low probability, some nodes may not receive a given message.
In that case, they may later ask for a retransmission of a transaction or endorsements related to a transaction.
After long failures (such as power outage or network partition), some nodes may have missed a large number of messages and become out-of-sync with the remainder of the network.
At this point, retransmitting every message becomes prohibitively expensive: that's why each node is able to synchronize its complete state from its neighbors.
We rely on the web of trust (\autoref{sec:wot}) to retrieve the state from neighbors that are sufficiently trusted by the out-of-sync node.
(In our implementation, a configurable quorum of identical values must be received before re-synchronizing one node's state.)

\subsection{Binary Veto Procedure}\label{sec:bvp_impl}

The main issue with our endorsement scheme is that Byzantine nodes can arbitrarily delay their endorsements.
To cope with that limitation in a practical way, we propose a BVP implementation in Algorithm~\ref{bvp}, based on periodic health probes of the gossip mechanism in our eventually synchronous network.

\begin{definition}\label{def:bound}
  The \emph{maximum gossip broadcast latency}, denoted $\tau$, is the maximum possible delay from a message broadcast to its delivery by \emph{every} correct node.
\end{definition}

We make the following two assumptions:
every correct node $p_{i}$ is able to estimate ($\mathcal{A}$) $\hat{\tau}$ such as $\hat{\tau} \geq \tau$ and ($\mathcal{B}$) $\delta_{i,j}$ the relative clock deviation for \emph{any} endorser $p_{j}$.
In practice, it is possible to obtain these two values from passive round-trip measurements in the gossip network.
(We note that under asynchrony, $\tau = \delta_{i,j} = \infty$.)
With that additional knowledge, each correct node can estimate locally the earliest possible sending time of a message, and discard the messages published after a specific deadline (line~\algref{bvp}{discard}).
This simple approach is sound during periods of synchrony, but may introduce significant delays due to use of a deadline. As BVP is not the main contribution of this paper, we leave the optimization of this primitive to future work.

\algnewcommand{\IIf}[1]{\State\algorithmicif\ #1\ \algorithmicthen}
\algnewcommand{\IElse}[1]{\State\algorithmicelse\ }

\begin{algorithm}[tb]
  \caption{Byzantine Veto Procedure (BVP) at node $p_{i}$}\label{bvp}
\begin{algorithmic}[1]
  \Function{BVP}{$\bar{T}, c_{i}$}
    \If{$c_{i} = 0$}\Comment{$p_i$ is vetoing the decision to drop $\bar{T}$}
      \State\textbf{return 0} \label{veto_fast}
    \EndIf{}

    \State\textbf{wait until} $\hat{\tau} < \infty \wedge \forall{j}, \delta_{i,j} < \infty$ \Comment{Wait for synchrony}\label{synchrony}
    \State$\text{deadline} = \max( t.d, t \in \bar{T} ) + \max(\delta_{i,j})$ \label{deadline}
    \State Stop delivering endorsements for $t \in \bar{T}$ sent \textbf{after} $\text{deadline}$ \label{discard}
    \Wait{}
      \State$\exists t \in \bar{T} : \textsc{Applicable}(t)$ \textbf{then return} 0\label{veto_slow}
      \State$now() > \text{deadline} + \hat{\tau}$ \textbf{then return} 1 \label{timeout}
    \EndUpon{}
  \EndFunction{}
\end{algorithmic}
\end{algorithm}

\begin{proposition}
  Algorithm~\ref{bvp} satisfies the properties of BVP.
\end{proposition}

\begin{proof}
  \textbf{\emph{Termination}} is trivial in eventually synchronous networks (line~\algref{bvp}{synchrony}).
  Per assumptions $\mathcal{A}$ and $\mathcal{B}$ every correct node will compute the same value for `$\text{deadline}$' at line~\algref{bvp}{deadline}.
  By line~\algref{bvp}{discard}, no endorsement for $t \in \bar{T}$ sent after this shared $\text{deadline}$ can be accepted.
  Thanks to assumption $\mathcal{A}$, endorsements sent before `$\text{deadline}$' are delivered before $(\text{deadline} + \hat{\tau}) > (\text{deadline} + \tau)$,
  leading to the same set of endorsements for $\bar{T}$ being received for every correct node after ($\text{deadline} + \hat{\tau}$).
  This implies the \textbf{\emph{Agreement}} property given the decisions of lines~\algref{bvp}{veto_slow} and~\algref{bvp}{timeout}.
  A correct node proposes a veto \emph{if and only if} at least one transaction in $\bar{T}$ is \textsc{Applicable}: \textbf{\emph{Validity}} follows from the properties of \textsc{Applicable}.
\end{proof}

\section{Evaluation}\label{sec:evaluation}

\subsection{Experimental setup}

We tested our implementation of \proj{} in two different environments: an emulation setup, and a global network using Amazon Web Services (AWS).
The emulation was performed on a server able to sustain several hundreds of nodes with the Mininet~\cite{MininetTeam} network emulation tool (48 threads of Intel(R) Xeon(R) Gold 6136 CPU at 3.00GHz with 188GB of RAM).
We used Mininet~\cite{MininetTeam} to isolate nodes from each other and to simulate real network latencies.
We drew latency values from an exponential distribution law with an average of $20$ ms per link.
Every node's clock was shifted by a random amount in the $[-5, 5]$ seconds interval between the reference time to simulate a relatively small asynchrony between network participants.
For the BVP algorithm, we chose the conservative value $\hat{\tau} = 10$ seconds: this leads to a practical checkpoint timeout of 20 second.
Each experiment was run 40 times, taking the average as the final result.

\subsection{Baseline}

To compare our work with available solutions, we executed the same experiments with \BFTSMART v1.2 server~\cite{Bessani2014} and a Tendermint v0.32.5 voting application~\cite{Buchman2016}.
\BFTSMART has been recognized as an efficient Java library for the BFT problem, and is being added in a number of applications, including Hyperledger Fabric~\cite{Androulaki2018,Sousa2018}.
Tendermint is a BFT Consensus mechanism based on a permissioned blockchain with a leader-based algorithm; its implementation relies on a gossip broadcast primitive, like \proj{}.
Both implementations allow custom application logic to be executed during consensus; this empowered us to emulate a voting behavior within these two existing solutions.
The two systems are leader-based, but their consensus choices are quite different: while \BFTSMART rely on a single leader as long as it reports no issue to avoid costly view changes, Tendermint leaders are selected in a round-robin fashion with each leader batching transactions into blocks.
\BFTSMART is based on a fully connected mesh topology whereas Tendermint nodes communicate via gossip.
The two baselines offer a different trade-off than our proposal, targeting stronger consistency guarantees \emph{but with no native democratic capabilities}.
(For fair comparison, we configured Tendermint with the ``skip timeout commit'' option to optimize its commit latency.)

\subsection{Network size ($n$)}

\begin{figure}
  \centering
  \begin{tikzpicture}
  \pgfdeclareplotmark{--}{%
    \pgfpathmoveto{\pgfqpoint{-4pt}{0pt}}
    \pgfpathlineto{\pgfqpoint{4pt}{0pt}}
    \pgfusepathqstroke
  }
  \tikzset{every mark/.append style={scale=0.7}}
  \begin{axis}[
      small,
      ybar,
      bar width=6pt,
      height=0.23\textwidth,
      width=0.45\textwidth,
      major grid style={densely dotted},
      ymajorgrids=true,
      legend columns=4,
      legend style={at={(1.05, 1.05)},anchor=south east,font=\scriptsize,draw=none},
      xlabel={Number of endorsers $n$},
      ylabel={$95^{th}$ percentile\\latency (s)},
      ylabel style={align=center},
      xtick distance=1,
      ytick distance=1,
      xticklabel={\pgfmathparse{10 * 2^(\tick - 1)}\pgfmathprintnumber{\pgfmathresult}},
      ymin=0,
      ymax=4,
    ]

    \addplot [red,pattern=north east lines,pattern color=red]
      table [x expr=\lineno,y expr=\thisrow{l95}/1000,col sep=comma] {varying_n_bftsmart_csv.tex};
    \addlegendentry{\BFTSMART}

    \addplot [green!60!black,pattern=crosshatch dots,pattern color=green!80!black]
      table [x expr=\lineno,y expr=\thisrow{l95}/1000,col sep=comma] {varying_n_tendermint_csv.tex};
    \addlegendentry{Tendermint}

    \addplot [blue,fill=blue!30!white]
      table [x expr=\lineno,y expr=\thisrow{l95}/1000,col sep=comma] {varying_n_ours_csv.tex};
    \addlegendentry{\proj{} (ours)}

    \addplot [mark=--,line width=1pt,black,sharp plot,update limits=false,draw=none,line legend,xshift=-8pt]
      table [x expr=\lineno,y expr=\thisrow{l50}/1000,col sep=comma] {varying_n_bftsmart_csv.tex};
    \addplot [mark=--,line width=1pt,black,sharp plot,update limits=false,draw=none,line legend]
      table [x expr=\lineno,y expr=\thisrow{l50}/1000,col sep=comma] {varying_n_tendermint_csv.tex};
    \addplot [mark=--,line width=1pt,black,sharp plot,update limits=false,draw=none,line legend,xshift=8pt]
      table [x expr=\lineno,y expr=\thisrow{l50}/1000,col sep=comma] {varying_n_ours_csv.tex};
    \addlegendentry{Median latency}

    \addplot [red,mark=text,text mark={$\oslash$},draw=none,sharp plot,update limits=false,xshift=-6pt] coordinates {(4, 0.2) (5, 0.2) (6, 0.2)};
    \addplot [green!60!black,mark=text,text mark={$\oslash$},draw=none,sharp plot,update limits=false] coordinates {(6, 0.2)};

    \node [anchor=north west,font=\scriptsize] at (axis cs:5.05,4) {$\uparrow$16.0};
    \draw [->,black] (rel axis cs:0.04,0.95) -- node [anchor=west,font=\scriptsize] {Better} (rel axis cs:0.04, 0.7);

  \end{axis}
\end{tikzpicture}
  \caption{Single transaction commit latency with increasing number of endorsers nodes ($n$) and emulated WAN latencies.
  The $\oslash$ symbols mean that we were unable to perform the experiment for a specific $n$ due to network contention.
  \proj{} clearly offers best network scalability.\label{fig:varying_n}}
\end{figure}
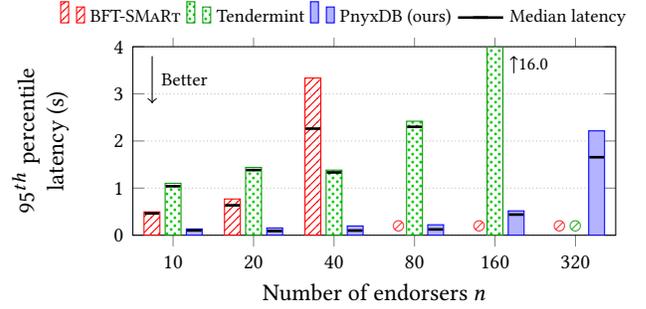

This first experiment measures the latency from a single transaction submission to its commit by every node.
We set the required number of endorsers to $\textstyle\omega = \left \lfloor \frac{2}{3}n \right \rfloor + 1$ and increased $n$ from 10 to 320.
For completeness, we note that setting $\omega = n$ (unanimous agreement) had the effect of slightly increasing the latency,
since nodes had to wait for more votes before committing any transaction.
As denoted by the $\oslash$ symbols, we were unable to complete some large network experiments for \BFTSMART ($n \geq 80$) and Tendermint ($n \geq 320$) in our testbed, due to extremely high CPU and network load.
\autoref{fig:varying_n} shows that \proj{} outperforms existing implementations for small and large networks by an order of magnitude.

\subsection{Number of clients}

To measure the effect of client scaling, we configured a various number of clients to submit transactions at an average rate of 2 transactions per second, as controlled by a Poisson point process.
The transactions were generated using the Yahoo! Cloud Serving Benchmark (YCSB)~\cite{Cooper2010}, a well recognized non-relational datastore testing tool that allowed us to vary the ratio of conflicting transactions, and hence the contention level on \proj{}.
We customized the benchmark workload to create only \textit{update} transactions to a set of 100 keys, from which updated keys were selected using a uniform distribution.
This relatively low level of contention reflects a number of real workloads, but we present some results for higher contention rates in~\autoref{sec:contention_effect}.
Additionally, each tested network required a quorum of $\omega = 7$ endorsers among $n = 10$ to tolerate at most $f = 3$ faulty nodes.

The transaction commit latencies and throughput are shown in \autoref{fig:varying_clients}.
While Tendermint and \proj{} were able to deal with up to 30 transaction commits per second, \BFTSMART was quickly saturated with client transactions: this is due to the large number of messages emitted during the successive rounds of consensus, and our realistic setup with realistic network latencies.
\proj{} performed well for the very large majority of transactions, providing an order of magnitude of latency improvement compared to Tendermint, and approached the optimal throughput while ensuring a low number of dropped transactions.
As summarized in~\autoref{table:evaluation}, \BFTSMART ensured that no single transaction was dropped.
However, Tendermint nodes were unable to commit around 9.3\% of transactions:
from our understanding, some nodes failed to keep their state synchronized with the network and gave up processing transactions.
\proj{} experienced less than 2.3\% of transaction drop on average.

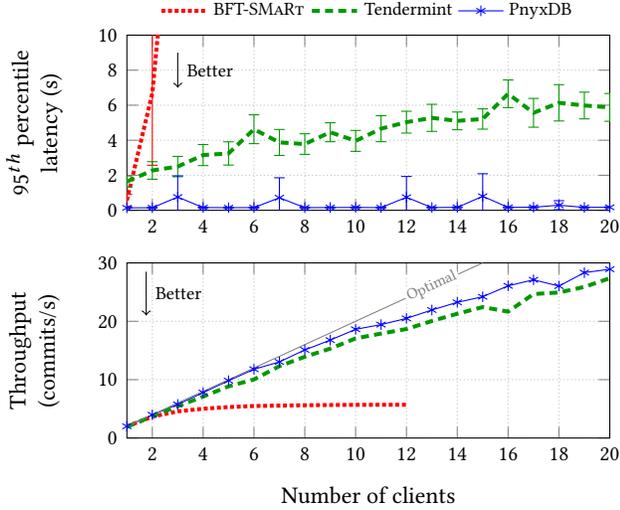
\begin{figure}[tb]
  \centering
\pgfplotsset{select coords between index/.style 2 args={
    x filter/.code={
        \ifnum\coordindex<#1\def\pgfmathresult{}\fi
        \ifnum\coordindex>#2\def\pgfmathresult{}\fi
    }
}}
\tikzstyle{every pin}=[fill=white,draw=gray,font=\footnotesize\color{gray}]

\begin{tikzpicture}

  \begin{axis}[
      name=latency,
      enlarge x limits=false,
      small,
      height=0.22\textwidth,
      width=0.45\textwidth,
      major grid style={densely dotted},
      xmajorgrids=true,
      ymajorgrids=true,
      legend columns=3,
      legend style={at={(0.5, 1.25)},anchor=north,font=\scriptsize,draw=none},
      xmin=1,
      xmax=20,
      xtick distance=2,
      ylabel={$95^{th}$ percentile\\latency (s)},
      ylabel style={align=center},
      ymin=0,
      ymax=10,
      error bars/error bar style=solid,
      mark options={solid,scale=1},
    ]

    \addplot [red,line width=1.5pt,style=densely dotted,update limits=false,
    error bars/.cd,y dir=both,y explicit]
    table [x=clients,y expr=\thisrow{l95}/1000,y error expr=\thisrow{err}/1000,col sep=comma]
    {varying_clients_bftsmart_csv.tex};
    \addlegendentry{\BFTSMART}

    \addplot [green!60!black,line width=1.5pt,style=densely dashed,update limits=false,
    error bars/.cd,y dir=both,y explicit]
    table [x=clients,y expr=\thisrow{l95}/1000,y error expr=\thisrow{err}/1000,col sep=comma]
    {varying_clients_tendermint_csv.tex};
    \addlegendentry{Tendermint}

      \addplot+ [blue,mark=asterisk,error bars/.cd,y dir=both,y explicit]
    table [x=clients,y expr=\thisrow{l95}/1000,y error expr=\thisrow{err}/1000,col sep=comma]
    {varying_clients_ours_csv.tex};
    \addlegendentry{\proj{}}

    \draw [->,black] (axis cs:3,9) -- node [anchor=west,font=\scriptsize] {Better} (axis cs:3,7);
  \end{axis}

  \begin{axis}[
      name=throughput,
      at={($(latency.south)-(0,0.7cm)$)},
      anchor=north,
      enlarge x limits=false,
      small,
      height=0.22\textwidth,
      width=0.45\textwidth,
      major grid style={densely dotted},
      xmajorgrids=true,
      ymajorgrids=true,
      xlabel={Number of clients},
      xmin=1,
      xmax=20,
      xtick distance=2,
      ymin=0,
      ymax=30,
      ylabel={Throughput\\(commits/s)},
      ylabel style={align=center},
      mark options={solid,scale=1},
    ]

    \addplot+ [draw=blue,mark=asterisk,error bars/.cd,y dir=plus,y explicit]
    table [x=clients,y=ours,col sep=comma]
    {throughput_csv.tex};

    \addplot [red,line width=1.5pt,style=densely dotted,update limits=false,
    select coords between index={0}{11}]
    table [x=clients,y=bftsmart,col sep=comma]
    {throughput_csv.tex};

    \addplot [green!60!black,line width=1.5pt,style=densely dashed,update limits=false,
    error bars/.cd,y dir=both,y explicit]
    table [x=clients,y=tendermint,col sep=comma]
    {throughput_csv.tex};

    \addplot+ [mark=none,draw=gray,domain=0:15,empty legend] {2*x};
    \node [gray,fill=white,rotate=24.75,inner sep=1pt,font=\tiny] at (axis cs:13,26) {Optimal};
    \draw [->,black] (rel axis cs:0.04,0.95) -- node [anchor=west,font=\scriptsize] {Better} (rel axis cs:0.04, 0.7);
  \end{axis}
\end{tikzpicture}
  \caption{{} Commit latency and  throughput with increasing load and contention, each client submitting 2 transactions per second from 100 records selected by YCSB. \proj{} can scale with the number of clients while offering best latencies and good throughput ($n=10$ and $\omega=7$).}\label{fig:varying_clients}
\end{figure}

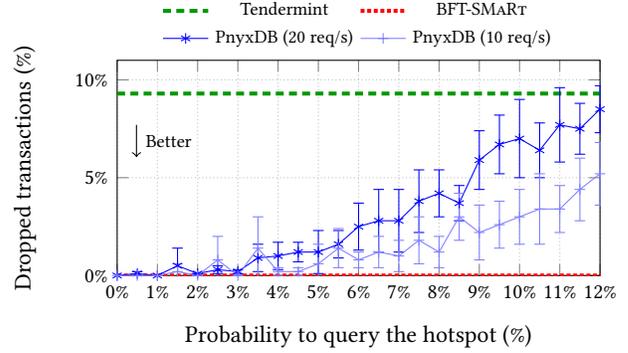
\begin{figure}[tb]
  \centering
  \tikzstyle{every pin}=[fill=white,draw=gray,font=\footnotesize\color{gray}]
\begin{tikzpicture}
  \begin{axis}[
      name=hotspot,
      enlarge x limits=false,
      small,
      height=0.25\textwidth,
      width=0.45\textwidth,
      major grid style={densely dotted},
      xmajorgrids=true,
      ymajorgrids=true,
      legend columns=2,
      legend style={at={(0.5, 1.01)},anchor=south,font=\scriptsize,draw=none},
      xlabel={Probability to query the hotspot (\%)},
      xticklabel={\pgfmathprintnumber{\tick}\%},
      xmin=0,
      xmax=12,
      xtick distance=1,
      ylabel={Dropped transactions (\%)},
      yticklabel={\pgfmathprintnumber{\tick}\%},
      ymin=0,
      ymax=11,
    ]

    \addplot [green!60!black,line width=1.5pt,style=densely dashed,update limits=false] coordinates {(0, 9.3) (20, 9.3)};
    \addlegendentry{Tendermint}

    \addplot [red,line width=1.5pt,style=densely dotted,update limits=false] coordinates {(0, 0) (20, 0)};
    \addlegendentry{\BFTSMART}

    \addplot+ [blue,mark=asterisk,error bars/.cd,y dir=both,y explicit]
    table [x=hotspot,y expr=(1000-\thisrow{committed_queries})/10,y error expr=\thisrow{err}/10,col sep=comma]
    {varying_hotspot_ours_csv.tex};
    \addlegendentry{\proj{} (20 req/s)}

    \addplot+ [blue!50!white,mark=+,error bars/.cd,y dir=both,y explicit]
    table [x=hotspot,y expr=(500-\thisrow{committed_queries})*0.2,y error expr=\thisrow{err}*0.2,col sep=comma]
    {varying_hotspot_ours_1s_csv.tex};
    \addlegendentry{\proj{} (10 req/s)}

    \draw [->,black] (rel axis cs:0.04,0.7) -- node [anchor=west,font=\scriptsize] {Better} (rel axis cs:0.04, 0.55);


  \end{axis}


\end{tikzpicture}
  \caption{Drop rate analysis for increasing YCSB hotspot selection ratio. As expected for \proj{}, the drop rate increases with the hotspot contention.}\label{fig:varying_hotspot}
\end{figure}

\begin{table*}
  \centering
  \captionsetup{justification=centering}
  \caption{Summary of comparison with emulated network latencies and 10 clients for a total of 1000 transactions.\\
  This is the average over 40 experiments with 10 nodes tolerating up to 3 Byzantine faults ($n = 10, \omega = 7$).}\label{table:evaluation}
  {\footnotesize
  \begin{tabular}{r|cc|cc|cc|c|c|c|c|c}
  \toprule
    & \multicolumn{2}{|c|}{\textbf{Average}} & \multicolumn{2}{|c|}{\textbf{95$^{th}$ perc.}} &
    \multicolumn{2}{|c|}{\textbf{Throughput}} & \textbf{Drop} & \textbf{Disk usage} & \textbf{Transfer} & \textbf{Exp.} & \textbf{Bandwidth per node} \\

    & \multicolumn{2}{|c|}{\textbf{latency}} & \multicolumn{2}{|c|}{\textbf{latency}} &
    \multicolumn{2}{|c|}{} & \textbf{rate} & \textbf{per node} & \textbf{per node} & \textbf{duration} & \textbf{average / max} \\[\medskipamount]

    \BFTSMART~\cite{Bessani2014}  & 89 s              &        & 170 s             &        & 5.68 tx/s             &      & \textbf{0\%}  & -               & 36 MB          & 230 s & 0.16 / 0.21 MB/s\\
    Tendermint~\cite{Buchman2016} & \underline{1.7 s} &        & \underline{3.9 s} &        & \underline{17.0 tx/s} &      & 9.3\%         & 26 MB           & 26 MB          & 65 s  & 0.40 / 1.40 MB/s\\
    \proj{} (ours)                & \textbf{0.15 s}   &$\div$11& \textbf{0.16 s}   &$\div$24& \textbf{18.6 tx/s}    &+9.4\%& 2.3\%         & \textbf{1.4 MB} & \textbf{20 MB} & 71 s  & 0.28 / 1.27 MB/s\\
  \bottomrule
  \end{tabular}
  }
\end{table*}

\subsection{Effect of contention}\label{sec:contention_effect}

We increased the level of transaction contention by rising a ``hotspot'' hit probability (\autoref{fig:varying_hotspot}), one option provided by YCSB.
This parameter has an immediate effect on the probability that (at least) two transactions ask to update the \textit{same} datastore record around the \textit{same} time, thus becoming conflicting transactions for \proj{} in our setup.
(We also note that the artificially-added clock shift tend to increase the probability of conflicts with high contention levels.)
We kept 10 clients for those experiments, leading to an average of 20 transactions per second (tx/s).
We also tested \proj{} with a slower rate of 10 transactions per second for comparison.
The worst case scenario would be that every transaction hitting the hotspot is eventually dropped:
in \proj{}, a transaction is dropped if a significant number of nodes ($n - \omega + 1$) are unable to endorse the transaction before its deadline.

As expected, the contention level has a direct impact on the ratio of dropped transactions.
For low levels of contention (up to 3\%), almost no transaction is being dropped thanks to conditional endorsements.
For higher levels of contention (up to excessively high levels), the drop ratio stays well below the worst case scenario.

\subsection{Speculative execution}\label{sec:speculative}
We implemented speculative execution (\autoref{sec:protocol}) to compare its latencies against classic commit latencies as presented in the previous subsections.
First, let us recall that once a transaction is \textsc{Applied} by speculative execution, there is no guarantee that it will eventually become \textsc{Committed}.
However, we can expect that the probability of rolling back to a non-\textsc{Applicable} (or \textsc{Dropped}) state stays very low (\autoref{fig:checkpoints}).

\autoref{fig:speculative_exec} shows the average latency observed for weak and strong guarantees with an increasing level of concurrency (resp. \textsc{Applied} and \textsc{Committed} states).
Speculative execution benefits are clearly visible for any level of contention, improving latency by up to 50\%,
a boost that could  clearly benefit applications that only require weak guarantees.
This improvement came with no additional cost, since \textit{at most} 0.03\% of speculative executions had to rollback.
In the classic configuration, the commit latency can be severely affected by pending checkpoints, leading to less-predictable performance.
As shown in~\autoref{fig:speculative_exec}, the commit latency is clearly more foreseeable in the speculative execution mode, staying stable despite increasing contention.
This observation confirms the positive impact of speculative execution.

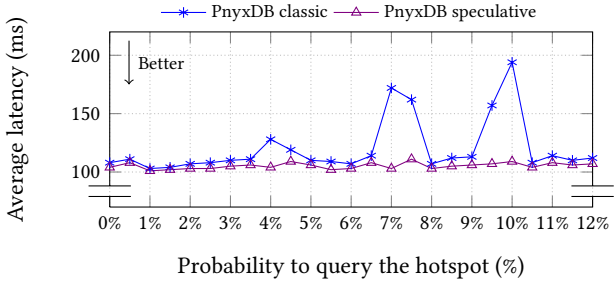
\begin{figure}
  \centering
  \begin{tikzpicture}
  \begin{axis}[
      enlarge x limits=false,
      small,
      height=0.22\textwidth,
      width=0.45\textwidth,
      major grid style={densely dotted},
      xmajorgrids=true,
      ymajorgrids=true,
      legend columns=2,
      legend style={at={(0.5, 1)},anchor=south,font=\scriptsize,draw=none},
      xlabel={Probability to query the hotspot (\%)},
      xticklabel={\pgfmathprintnumber{\tick}\%},
      xmin=0,
      xmax=12,
      xtick distance=1,
      ylabel={Average latency (ms)},
      ymin=70,
      ymax=220,
      axis y discontinuity=parallel,
    ]

    \addplot+ [draw=blue,mark=asterisk]
    table [x=hotspot,y expr=\thisrow{mean},col sep=comma]
    {varying_hotspot_ours_csv.tex};
    \addlegendentry{\proj{} classic}

    \addplot+ [draw=violet,mark=triangle]
    table [x=hotspot,y expr=\thisrow{mean},col sep=comma]
    {varying_hotspot_ours_applied_csv.tex};
    \addlegendentry{\proj{} speculative}

    \draw [->,black] (rel axis cs:0.04,0.95) -- node [anchor=west,font=\scriptsize] {Better} (rel axis cs:0.04, 0.7);
  \end{axis}
\end{tikzpicture}
  \caption{Difference between classic and speculative operation execution latencies. There is a significant gain with speculative execution for high levels of contention.}\label{fig:speculative_exec}
\end{figure}
\begin{figure}
  \centering
  \begin{tikzpicture}
  \begin{axis}[
      name=a,
      area style,
      enlarge x limits=false,
      small,
      height=0.18\textwidth,
      width=0.45\textwidth,
      major grid style={densely dotted},
      xmajorgrids=true,
      ymajorgrids=true,
      legend columns=2,
      legend style={at={(0.5, 1.2)},anchor=north,font=\scriptsize,draw=none},
      xticklabel={\ },
      yticklabel={\pgfmathprintnumber[fixed]{\tick}\%},
      scaled ticks=false,
      xmin=0,
      xmax=12,
      xtick distance=1,
      ylabel style={align=center,font=\footnotesize},
      ylabel={Rollback\\ratio (\%)},
      ymin=0,
    ]

    \addplot [fill=black!30,draw=black,mark=asterisk,error bars/.cd,y dir=both,y explicit]
    table [x=hotspot,y=rollbacks,y error=err,col sep=comma]
    {rollbacks_csv.tex} \closedcycle;

    \draw [->,black] (rel axis cs:0.04,0.95) -- node [anchor=west,font=\scriptsize] {Better} (rel axis cs:0.04, 0.7);
  \end{axis}

  \begin{axis}[
      name=b,
      at={($(a.south)-(0,0.2cm)$)},
      anchor=north,
      area style,
      enlarge x limits=false,
      small,
      height=0.18\textwidth,
      width=0.45\textwidth,
      major grid style={densely dotted},
      xmajorgrids=true,
      ymajorgrids=true,
      legend columns=2,
      legend style={at={(0.5, 1.2)},anchor=north,font=\scriptsize,draw=none},
      xlabel={Probability to query the hotspot (\%)},
      xticklabel={\pgfmathprintnumber{\tick}\%},
      xmin=0,
      xmax=12,
      xtick distance=1,
      ylabel style={align=center,font=\footnotesize},
      ylabel={Checkpoints},
      ymin=0,
    ]

    \addplot [fill=black!30,draw=black,mark=asterisk,error bars/.cd,y dir=both,y explicit]
    table [x=hotspot,y=checkpoints,y error=err,col sep=comma]
    {checkpoints_csv.tex} \closedcycle;

    \draw [->,black] (rel axis cs:0.04,0.95) -- node [anchor=west,font=\scriptsize] {Better} (rel axis cs:0.04, 0.7);
  \end{axis}
\end{tikzpicture}
  \caption{Top: experimental ratio between the number of rollbacks and the number of transaction applications. Bottom: number of checkpoints executed across 1000 transaction submissions.}\label{fig:checkpoints}
\end{figure}
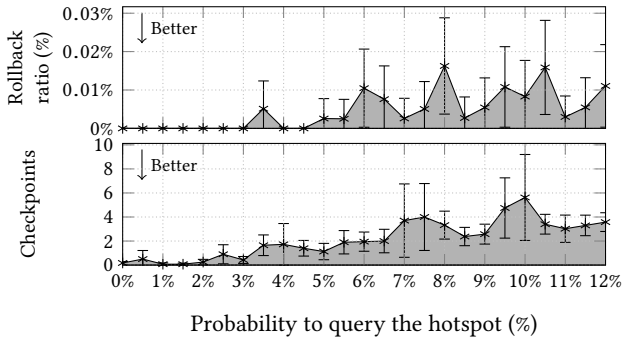

\subsection{Impact of checkpoints for conflict resolution}\label{sec:checkpoints_eval}

Algorithms~\autoref{checkstatealgo} and~\autoref{checkpointalgo} suggest that checkpoints must be triggered immediately when a single transaction could be dropped by a node.
This is inefficient, since the proposed checkpoint procedure is costly in terms of bandwidth.
Thus, we added a pooling mechanism to limit the number of checkpoints: by aggregating transactions before proposing checkpoints, each node optimizes its bandwidth while slightly increasing the commit latencies of conflicting transactions.
In our evaluation, at most 10 checkpoints were triggered (\autoref{fig:checkpoints}).
Given that checkpoints happen mainly in times of high levels of contention, we can conclude that their number is still practical, thanks to our pooling optimization.

The longest path measured in the graphs of conditions was of length 34, requiring less than 1 millisecond to be processed.
This observation indicates that our conditional endorsement scheme is scalable and practical in terms of complexity.
We note that non-conflicting transactions were \textit{not} affected and continue to be committed even when the hotspot probability is high.
This is not the case for the other baselines, where transactions are totally ordered by successive leaders.

\subsection{Large scale experiments}

\begin{table}[t]
  \caption{Worldwide AWS deployment: inter-region round-trip time (May 6, 2019).
    Nodes were evenly sharded between regions.}\label{fig:aws_latencies}
  {\footnotesize
  \begin{tabular}{r|c|c|c|c|c}
  \toprule
    \emph{(in ms)}  & Virginia & São Paulo & Paris & Frankfurt & Sydney \\
    California & 61 & 194 & 138 & 147 & 149 \\
    Virginia   &    & 147 & 79  & 88  & 204 \\
    São Paulo  &    &     & 223 & 226 & 315 \\
    Paris      &    &     &     & 10  & 283 \\
    Frankfurt  &    &     &     &     & 292 \\
  \bottomrule
  \end{tabular}
  }
\end{table}

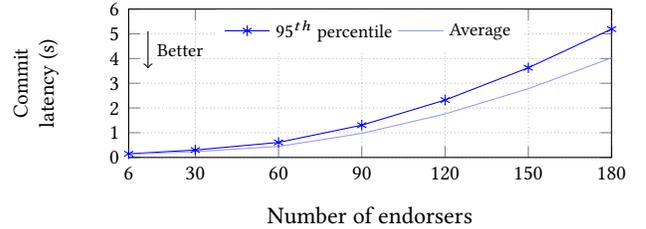
\begin{figure}
  \begin{tikzpicture}
  \begin{axis}[
      small,
      height=0.20\textwidth,
      width=0.45\textwidth,
      major grid style={densely dotted},
      xmajorgrids=true,
      ymajorgrids=true,
      legend columns=2,
      legend style={at={(0.5, 1.)},anchor=north,font=\scriptsize,draw=none},
      xlabel={Number of endorsers},
      xtick=data,
      xmin=6,
      xmax=180,
      ylabel={Commit\\latency (s)},
      ylabel style={align=center,font=\footnotesize},
      ymin=0,
      ymax=6,
      ytick distance=1,
    ]

    \addplot+ [draw=blue,mark=asterisk]
    table [x=n,y expr=\thisrow{l95}/1000,col sep=comma] {aws_csv.tex};
    \addlegendentry{$95^{th}$ percentile}

    \addplot+ [draw=blue!40!white,mark=none]
    table [x=n,y expr=\thisrow{mean}/1000,col sep=comma] {aws_csv.tex};
    \addlegendentry{Average}

    \draw [->,black] (rel axis cs:0.04,0.85) -- node [anchor=west,font=\scriptsize] {Better} (rel axis cs:0.04,0.6);
  \end{axis}
\end{tikzpicture}
  \caption{Worldwide AWS deployment: commit latency with increasing number of endorser nodes. $\omega$ was set to 70\% of endorsers.\label{fig:aws}}
\end{figure}

To assert \proj{}'s scalability in a global setting, we used Amazon Web Services (AWS) \texttt{t2.micro} EC2 instances to deploy nodes uniformly in 6 AWS regions (\autoref{fig:aws_latencies}).
During our experiments, we observed a steady maximum inter-region round-trip time of 315 ms between São Paulo and Sydney regions.
Even under these conditions, \proj{} managed to commit transactions with an average latency lower than 2 seconds for networks up to 180 nodes (\autoref{fig:aws}).

\section{Related work}\label{sec:related_work}

There exists a large number of consensus proposals for block\-chain-like applications \cite{Garay2018,Bano}. The consensus problem is tackled very differently in public permissionless systems compared to permissioned consortium systems.
For public systems, there have been many efforts towards Proof-Of-Stake consensus:
a small subset of participants are pseudo-randomly selected to lead the consensus for a specific round, based on their account balance (i.e. \textit{stake} in the network).
Recent proposals rely on a multiparty coin-flipping protocol for leader selection \cite{Kiayias2017}, or propose a probabilistic method using verifiable random functions~\cite{Gilad2017}.
These proposals are said to be easily vulnerable to Sybil attacks~\cite{Douceur2002} since anyone can participate.
Proof-Of-Work schemes are still considered to be the safest, and main public networks still use it extensively.
More efficient algorithms have been proposed to provide more fairness to small devices~\cite{Popov2018}, less communication rounds~\cite{Abraham2018} or more useful computations~\cite{Ball2017}.

In this paper we focus on permissioned systems where participants are known in advance.
This path allows for more flexibility in the choice of the threat model and the trust assumptions.
RSCoin~\cite{Danezis2015} relies on a trusted central bank and on distributed sets of authorities for improved scalability.
Similarly, the Corda platform~\cite{Hearn2016} puts trust in small notary clusters running consensus algorithms like \BFTSMART~\cite{Bessani2014}. 
In its current version, Hyperledger Fabric~\cite{Androulaki2018,Istvan2018} also requires that a centralized ordering service is trusted by every party.
These solutions, while giving good scalability promises, rely on central points of trust, that if manipulated by a malicious actor would break the entire sytem.
The common assumption is that such entities would be legally accountable through audits:
to remove that assumption, Sousa et.\ al.~\cite{Sousa2018} proposed to replace the current Kafka ordering service by \BFTSMART in Hyperledger Fabric.
\proj{} leverages a web of trust to ensure good scalability and node recovery, while avoiding central points and elected leaders.
We believe this makes our proposal more robust to corruption and malicious manipulation.

Other consortium systems have also been proposed~\cite{Buchman2016,Crain2017,Friedman2017,Rocket2019}.
Randomization techniques have been used to solve asynchronous BFT consensus~\cite{Moniz2006,Moniz2010,Miller2016}, among which BEAT~\cite{Duan2018} that suggests to rely on recent cryptographic primitives.
Such systems usually require that correct nodes present deterministic execution for consistency~\cite{Cachin2016,Yamashita2019}.
By comparison, \proj{} relies on a consensus algorithm specifically designed for non-deterministic democratic decisions, and exploits operations commutativity in a similar way than~\cite{Singh2009,Raykov2011,Hearn2016,Park2019}.
In similar AWS deployments, BEAT reports a latency of around 500 ms for 6 nodes and more than 1 minute for 92 nodes, while \proj{} proposes 130 ms and less than 1 second, respectively~\cite{Duan2018}.

Speculative execution has been proposed in BFT consensus to reduce latencies~\cite{Kotla2009,Singh2009,Garcia2011,Luiz2011,Duan2014}.
We considered this optimization to speed-up our state synchronization algorithm while achieving extremely low rates of rollbacks.
Finally, note that some vote schemes have been proposed \cite{Keleher1999,Barreto2007}, but they apply only for non-Byzantine fault models.

\balance
\section{Discussion}\label{sec:discussion}

\proj{} has been designed to work with state-of-the-art networking techniques.
However, some elements can affect its liveness.

  \textbf{Invalid deadline} A client may submit a transaction with a very low or high deadline relative to the absolute time.
    The first case is handled by the endorsement conditions mechanism, but nodes may block in the second case.
    Bounds on deadlines would be a simple countermeasure to filter incoming transactions and avoid endorsing transactions with out-of-bounds deadlines~\cite{Clement2009}.

  \textbf{Conflicting transaction flooding} A rogue client could send many simultaneous conflicting transactions, such that it will be hard to reach a single quorum agreement within the transactions deadline.
    This attack will not break the safety, but the system may drop transactions, with a large number of checkpoints being handled.
    A solution would be to isolate the responsible node and rate-limit it.

  \textbf{Query drops} As underlined in our evaluation, query drops are expected during contention.
    This behavior is very common in classic and BFT replicated databases~\cite{Luiz2011,Pedone2011}, and each client could easily propose several transactions until one is finally committed.
    It is however clear that \proj{} has been designed mainly for commutative workloads, as commonly seen in modern distributed applications.

  \textbf{Checkpoint with asynchrony} Our BVP implementation waits for good network conditions before allowing dropping transactions.
  This is a safety requirement, given that Byzantine nodes could delay their endorsements indefinitely under asynchrony.
  An interesting property is that non-conflicting transactions are always allowed to proceed, independently of pending checkpoints.

  \textbf{Other optimizations} We did not test batching of transactions to increase throughput~\cite{Rocket2019,Sousa2018,Singh2009,Kotla2009,Duan2018}.
    We focused in this work on latency optimization, hence we believe that transaction batching is an orthogonal optimization that may further increase \proj{} throughput.

\section{Conclusion}\label{sec:conclusion}

In this paper, we presented \emph{\proj{}}, a scalable eventually-consistent BFT replicated datastore.
At its core lies a scalable low-latency conflict resolution protocol, based on \emph{conditional endorsements}.
\proj{} supports nodes having different beliefs and policy agendas, allowing to build new kinds of democratic applications with first-class support for non-conflicting transactions.
Compared to popular BFT implementations, we demonstrated that our system is able to reduce commit latencies by an order of magnitude under realistic conditions, while ensuring steady commit throughput.
In particular, our experimental evaluation shows that \proj{} is capable of scaling to up to hundreds of replicas on a geodistributed cloud deployment.


\newgeometry{twoside=true, head=13pt,
twocolumn, paperwidth=8.5in, paperheight=11in,
includehead=false, columnsep=2pc,
top=1.1in, bottom=1.1in, inner=0.75in, outer=0.75in,
}
\bibliographystyle{ACM-Reference-Format}
\bibliography{sources.bib,web.bib}
\end{document}